\documentclass[lettersize,onecolumn]{IEEEtran}
\usepackage{setspace} 
\usepackage{cite}
\usepackage{graphicx}
\usepackage{caption}
\usepackage{subcaption} 
\usepackage[cmex10]{amsmath}
\usepackage{mathtools}
\usepackage{amsthm}
\usepackage{amsfonts}
\usepackage[font=small,skip=0pt]{caption}
\usepackage{amssymb}
\usepackage{textcomp}
\usepackage{xcolor}
\def\BibTeX{{\rm B\kern-.05em{\sc i\kern-.025em b}\kern-.08em
    T\kern-.1667em\lower.7ex\hbox{E}\kern-.125emX}}
\usepackage{algorithm}
\usepackage{algpseudocode}
\usepackage{bm,xstring}
\usepackage{fixltx2e}
\usepackage{tikz}
\usepackage{hyperref}
\usetikzlibrary{matrix}
\usetikzlibrary{shapes,arrows,positioning}
\usetikzlibrary{decorations.pathreplacing}

\DeclareMathOperator*{\argmin}{argmin}
\usepackage{color}

\def\blue{\color{blue}}

\newcommand{\ignore}[1]{}

\newtheorem{example}{Example}
\newtheorem{lemma}{Lemma}

\newtheorem{theorem}{Theorem}

\theoremstyle{definition}

\newtheorem{definition}{Definition}
\newtheorem{assumption}{Assumption}
\newtheorem{remark}{Remark}

\begin{document}

\title{Contextual Status Updating: How to Maximize Situational Awareness?}




\title{Remote Safety Monitoring: Status Updating for Situational Awareness Maximization}

\title{Remote Safety Monitoring: Significance-Aware Status Updating for Situational Awareness}

\author{Tasmeen~Zaman~Ornee,~\IEEEmembership{Member,~IEEE,}
        Md~Kamran~Chowdhury~Shisher,~ \IEEEmembership{Member,~IEEE,}

        Clement~Kam,~\IEEEmembership{ Member,~IEEE,}
        and~Yin~Sun,~\IEEEmembership{Senior Member,~IEEE}
\thanks{A part of this manuscript was presented in {\it IEEE MILCOM Workshop on QuAVoI}, Boston, MA, 2023 \cite{Ornee_MILCOM}.}
\thanks{T. Z. Ornee is with the Department of Electrical and Computer Engineering, The Ohio State University, Columbus, OH, 43210 USA (email: ornee.1@osu.edu).}
\thanks{M. K. C. Shisher is with the Department of Electrical and Computer Engineering, Purdue University, West Lafayette, IN, 47907 USA (email: mshisher@purdue.edu).}
\thanks{T. Z. Ornee and M. K. C. Shisher were Ph.D. students in the Department of Electrical and Computer Engineering at Auburn University, Auburn, AL, 36830 USA.}
\thanks{C. Kam is with the U.S. Naval Research Laboratory, Washington, DC 20375 USA (e-mail: clement.k.kam.civ@us.navy.mil).}
\thanks{Y. Sun is with the Department
of Electrical and Computer Engineering, Auburn University, Auburn, AL, 36830 USA e-mail:  (yzs0078@auburn.edu).}
\thanks{This work was supported in part by the NSF grant CNS-2239677 and NRL 6.1 Base Program.}
}




\maketitle

\begin{abstract}
In this study, we consider a problem of remote safety monitoring, where a monitor pulls status updates from multiple sensors monitoring several safety-critical situations.
Based on the received updates, multiple estimators determine the current safety-critical situations. Due to transmission errors and limited channel resources, the received status updates may not be fresh, resulting in the possibility of misunderstanding the current safety situation. In particular, if a dangerous situation is misinterpreted as safe, the safety risk is high. We study the joint design of transmission scheduling and estimation for multi-sensor, multi-channel remote safety monitoring, aiming to minimize the loss due to the unawareness of potential danger. We show that the joint design of transmission scheduling and estimation can be reduced to a sequential optimization of estimation and scheduling. The scheduling problem can be formulated as a Restless Multi-armed Bandit (RMAB), for which it is difficult to establish indexability. We propose a low-complexity Maximum Gain First (MGF) policy and prove it is asymptotically optimal as the numbers of sources and channels scale up proportionally, without requiring the indexability condition. We also provide an information-theoretic interpretation of the transmission scheduling problem. Numerical results show that our estimation and scheduling policies achieves higher performance gain over random selection with queue, randomized policy, and Maximum Age First (MAF) policy. 
\end{abstract}


\section{Introduction}
\IEEEPARstart{A}{wareness} of dangerous situations is of paramount importance  in safety-critical systems \cite{grau2017industrial}. For example, in health monitoring systems, precise tracking of the
glucose level or heart rate is crucial to facilitate prompt precautionary measures \cite{abdulmalek2022iot}. In disaster response systems, it is important to monitor landslides (e.g., downfalls of a large mass of ground, rock fragments and debris) in unstable areas where intense rainfalls, floods or earthquakes occur \cite{adeel2019survey}. Landslides in such areas might cause loss of lives and damage buildings. 
These scenarios highlight the need for safety monitoring systems to assess 
situational awareness about remote systems. Any misunderstanding of the situational awareness can lead to severe consequences.


One challenge to efficiently utilize the state information in real-time is the limited capacity of the communication medium. Conversely, depending on the surrounding situation, some processes can be more significant to monitor than others
For example, in autonomous driving, a self-driving car deciding to change lanes needs to prioritize real-time information about vehicles in the adjacent lane over those in its current lane. 
To address this issue, it is important to quantify the significance, or the value, of status updates, for safety monitoring.

In this paper, we answer this question for a pull-based status updating system where multiple sensors monitor the status of several safety-critical situations. A central scheduler requests updates from sensors. Due to transmission errors and limited channel resources, the received updates may not be fresh. One performance metric that characterizes data freshness is the \emph{Age of Information (AoI)} \cite{KaulYatesGruteser-Infocom2012}. Let $U(t)$ be the generation time of the freshest received observation by time $t$. AoI, as a function of $t$, is defined as $\Delta(t) = t- U(t)$ which exhibits a linear growth with time $t$ and drops down to a smaller value whenever a fresher observation is delivered. However, 
AoI only captures the timeliness of the information, but not its significance. Hence, relying solely on AoI-based decision making is not sufficient, particularly, in safety-critical scenarios where misunderstanding about the situation can lead to significant performance loss. To address this, we employ multiple estimators in the receiver that estimate the safety situations, and by incorporating a loss function, we measure the significance of the estimated status update for safety monitoring. One important property of the loss function is: If a dangerous situation is incorrectly estimated as safe, the loss will be significantly higher because of the high risk of damage. Conversely, if the safe situation is incorrectly estimated as dangerous, the loss will be small. This is because even if the estimate is incorrect, the risk of damage is small.

Next, we seek to answer the following research question: {\it how can we jointly design efficient safety situation estimators and a transmission scheduling policy to maximize situational awareness?} The contributions of this paper are summarized as follows:

\begin{itemize}

\item We show that the joint design of estimators and scheduling policy can be reduced to a sequential optimization of estimators and scheduling policy (see Section \ref{joint_design}). Given the estimators, the goal of the multi-agent, multi-channel scheduling problem is to minimize the performance loss due to situational awareness while satisfying a channel resource constraint. The formulated problem is a Restless Multi-armed Bandit (RMAB). 

\item We significantly reduce the state space of the scheduling problem by leveraging the sufficient statistic of the history. While belief states are commonly used as sufficient statistics in the literature \cite{Liu_TIT, Ouyang_TMC, Peter_ITC, Ouyang_Infocom, Javidi_TIT, Ansell_2003, Neely_2010, Murugesan_2012, chen2021scheduling, chen2022index}, they can still result in a high-dimensional state space.  We find that the latest received observation and its age also form a sufficient statistic for estimating each agent's safety situation (see Theorem \ref{thm_sufficient_statistic}). Consequently, we replace the belief MDP framework with a Markov Decision Process (MDP) that uses the latest received observation and its age as states. This alternative formulation is crucial because the belief MDP approach leads to a quadratic increase in the size of the state space with AoI \cite{Chen_2022, chen2022index, Ouyang_Infocom, Ouyang_TMC}, while our method demonstrates only a linear increase (see Section \ref{probelm_simplification}). This results in a substantial computational advantage.


\item Our characterization of the optimal estimator reveals an interesting connection between the loss of situational awareness and the concept of generalized conditional entropy. We prove that the expected loss due to the unawareness of potential danger given the latest received observation and its age is a generalized conditional entropy (see Lemma \ref{lemma_entropy}). This characterization also facilitates the design of an efficient scheduling policy. Using this information-theoretic analysis, we show that frequent updates are necessary when the received observations are near safety boundaries (high uncertainty, low situational awareness). Conversely, when observations are far from the boundaries (low uncertainty, high situational awareness), less frequent updates are sufficient (see Section \ref{discussions}).

\item We develop an asymptotically optimal Maximum Gain First policy (see Algorithm \ref{algo1} and Theorem \ref{optimality}) by solving the multi-sensor, multi-channel transmission scheduling problem which is formulated as an RMAB. We utilize constraint relaxation and the Lagrangian method to decompose the original problem into multiple separated Markov Decision Processes (MDPs) and solve each MDP by dynamic programming \cite{bertsekas2011dynamic}. Most of the prior works \cite{tripathi2019whittle, ornee2023whittle, Shisher_2024, Liu_TIT, chen2021scheduling} in RMAB have utilized Whittle index policy that requires an indexability condition to satisfy. In our problem, the indexability condition is difficult to establish due to (i) complicated state transitions, (ii) unreliable channels, and (iii) general loss functions. The
benefit of our Maximum Gain First policy is that no indexability condition
is required to satisfy. Our results hold for general loss functions and both reliable and unreliable channels.

\item Numerical results illustrate that our scheduling policy achieves significant performance gain compared to random selection with queue, randomized policy, and Maximum Age First policy (see Section \ref{simulation}). Because our policy utilizes the knowledge of the latest received signal, it illustrates good performance compared to the other policies that ignore this knowledge.
 
\end{itemize}
\section{Related Work}

\textbf{AoI in context-aware updating:} 
Several research papers studied information-theoretic measures to evaluate the impact of information freshness along with information content \cite{sun2018information, sun2019sampling, VoI_Kosta, wang2022framework, soleymani2016optimal, Chen_2022, Shisher2022, Shisher_2024, shisher2021age, shisher2023learning}. In \cite{sun2019sampling, VoI_Kosta, wang2022framework, soleymani2016optimal}, the authors employed Shannon's mutual information to quantify the information carried by received data messages regarding the current signal at the source and used Shannon’s conditional entropy to measure the uncertainty about the current signal. Based on the studies of \cite{sun2019sampling, wang2022framework, soleymani2016optimal}, the authors in \cite{Chen_2022} utilized Uncertainty of Information (UoI) by using the Shannon's conditional entropy. 
In \cite{Shisher2022, Shisher_2024, shisher2021age, shisher2023learning}, a generalized conditional entropy associated with a loss function $L$, or $L$-conditional entropy $H_L(Y_t|, \Delta (t), X_{t-\Delta(t)})$ was utilized, where $Y_t$ is the true state of the source and $X_{t-\Delta(t)}$ is the observed value. 

In addition, minimization of linear and non-linear functions of AoI has been extensively studied in literature \cite{sun2018information, sun2019sampling, Sun_TIT_2017, Bedewy_2021, Ornee2021, Wiener_TIT, Klugel_infocom, jsac_survey}. One limitation of AoI is that it only captures the timeliness of the information while neglecting the actual influence of the conveyed information. 
To address this issue, several performance metrics were introduced in conjunction with AoI \cite{zhong2018two, Ali_AoII, zheng2020urgency, yates2021age, holm2021freshness, Chen_2022, Uysal_QAoI_ISIT, VoI_Kosta}. In \cite{Ali_AoII}, the concept of Age of Incorrect Information (AoII) was introduced which is characterized as a function of both age and estimation error. In \cite{zhong2018two}, Age of Synchronization (AoS) was considered along with AoI to measure the freshness of a local cache. Urgency of Information (UoI) was proposed in \cite{zheng2020urgency} that captures the context-dependence of the status information along with AoI. Version AoI was introduced in \cite{yates2021age} which represents how many versions are outdated at the receiver compared to the transmitter. An AoI at Query (QAoI) metric was investigated in \cite{holm2021freshness}, \cite{Uysal_QAoI_ISIT} to capture the freshness only when required in a pull-based communication system. Value of Information (VoI),  defined by the Shannon mutual information was investigated in \cite{VoI_Kosta}. 


{\textbf{AoI in Sampling and Scheduling:} There exist numerous papers on AoI-based sampling and scheduling \cite{sun2019sampling, Chen_2022,  Shisher_2024, ornee2023whittle, xiong2022index, zou2021minimizing, chen2021scheduling, shisher2023learning, Pull_Ji, Ornee2021, Wiener_TIT, Ornee_SPAWC}. 
Authors in \cite{Pull_Ji} studied an AoI minimization problem under a pulling model that considers replicated requests to the server. 
In \cite{sun2019sampling}, sampling policies for optimizing non-linear AoI functions were studied. 
AoI minimization for single-hop networks was studied in \cite{kadota2018scheduling}. 
A joint sampling and scheduling problem to minimize monotonic AoI functions were considered in \cite{Bedewy_2021}. A Whittle index-based scheduling algorithm to minimize AoI for stochastic arrivals was considered in \cite{hsu2018age}. Authors in \cite{tripathi2019whittle} proposed a Whittle index policy for minimizing non-decreasing AoI functions. 
In \cite{Chen_2022}, the authors proposed a Whittle index-based scheduling policy to minimize the UoI modeled as Shanon entropy. Optimal scheduling policies for both single and multi-source systems were studied and a Whittle index policy was proposed for multi-source case in \cite{Shisher_2024}. 
The optimal sampling policies for Gauss-Markov processes were studied in \cite{Ornee2021, Wiener_TIT, Ornee_SPAWC} where the estimation error becomes a monotonic function of age in signal-agnostic scenarios and the associated problem for minimizing age-penalty functions were reported. 
A Whittle index policy for continuous-time Gauss-Markov processes for both signal-aware and signal-agnostic scenarios was reported in \cite{ornee2023whittle}. 
Besides Whittle index-based policies that require an indexability condition, non-indexable scheduling policies were also studied in \cite{xiong2022index, zou2021minimizing, chen2021scheduling, chen2022index, shisher2023learning}. In this paper, because of the complicated nature of state transition along with erasure channels, indexability is very difficult to establish. However, we provide a “Maximum Gain First Policy” developed in \cite{shisher2023learning, chen2022index}. The comparison of our model with \cite{shisher2023learning, chen2022index} is that our model considers erasure channels, signal observation, and general loss functions. 
Our scheduling policy is designed for a pull-based communication model where the scheduling decisions are based on the latest received observation and its AoI. We further proved that the developed policy is asymptotically optimal.

\begin{figure}
\centering
\includegraphics[width=9cm]{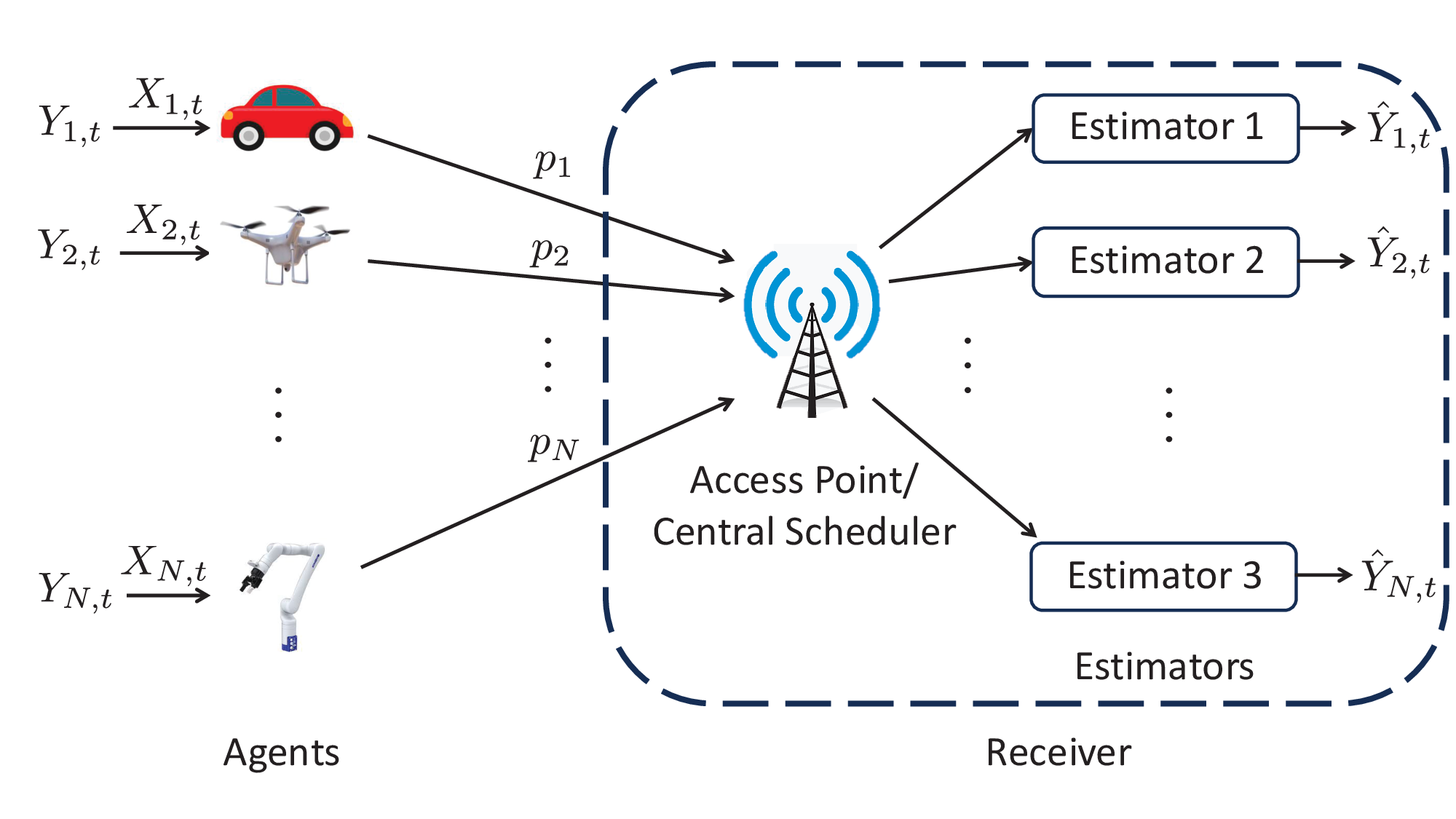}
\caption{A multi-agent, multi-channel safety monitoring system.}\vspace{-0.0cm}
\label{fig_model1}
\end{figure}

\textbf{Restless Bandits with Belief States:} RMAB problems are a well-established framework for studying sequential decision-making problems. Our problem focuses on an RMAB where each arm observes a finite state Markov process. There exists numerous closely related studies to our setting that focused on solving RMAB \cite{Liu_TIT, Ouyang_TMC, Ouyang_Infocom, Javidi_TIT, Ansell_2003, Neely_2010,  chen2021scheduling, chen2022index}. In \cite {Liu_TIT}, a Whittle index policy of a class of RMAB problems for dynamic multiaccess channels was studied. 
\cite{Ouyang_TMC} considered time-correlated fading channels with ARQ feedback. Opportunistic scheduling of multiple \emph{i.i.d.} channels is studied in \cite{Javidi_TIT}. 
In \cite{Neely_2010}, a multi-user wireless downlink has been studied for unknown channel states. 
\cite{Ansell_2003} developed an index heuristic for a multi-class queuing system with increasing convex holding cost rates. 
All of these studies tackle the problem by considering a belief MDP formulation (or POMDP formulation) using belief states (probability distribution over all possible states) \cite{Liu_TIT, Ouyang_TMC, Ouyang_Infocom, Neely_2010, Ansell_2003, chen2021scheduling, chen2022index}. Such formulations render the state space uncountable and leads to the curse of dimensionality \cite{Liu_TIT, Javidi_TIT, Neely_2010, Ansell_2003}. 
The difference between the formulation in \cite{Liu_TIT, Ouyang_TMC, Ouyang_Infocom, Javidi_TIT, Neely_2010, Ansell_2003, chen2021scheduling, chen2022index} and our problem is that we do not need to utilize belief states. By utilizing the sufficient statistic of the history observation (the latest received observation and the corresponding age value), we obtain a significantly smaller state space. 
Consequently, our framework is computationally more efficient compared to existing approaches \cite{Liu_TIT, Ouyang_TMC, Ouyang_Infocom, Javidi_TIT, Neely_2010, Ansell_2003, chen2021scheduling, chen2022index}.



\section{Model and Formulation}\label{model}

\subsection{System Model}
We consider the time-slotted, pull-based status-updating system, as depicted in Figure \ref{fig_model1}, where an access point retrieves the statuses 
of $N$ agents (e.g. cars, UAVs, robotic arms)
to monitor their safety levels (e.g., safe, cautious, dangerous). These statuses may include a car’s location on the road, images captured by a UAV-mounted camera, or the joint angles of a robotic arm in a factory. 
Let $X_{n, t} \in \mathcal X_n$ denote the status of agent $n$ at time $t$. 
We assume that $X_{n,t}$ follows a Markov chain and the processes \{$X_{n, t}, t = 0,1,2,\ldots$\} and \{$X_{m, t}, t = 0,1,2,\ldots$\} are independent for all $n\neq m$.
The safety level of agent $n$ is denoted as 
$Y_{n,t} \in \mathcal Y_n$.  We assume that $\mathcal{X}_n$ and $\mathcal{Y}_n$ are discrete and finite sets.
In addition, we make the following assumption: 
\begin{assumption}\label{assumption1}
As illustrated in Figure \ref{hidden_MC}, $Y_{n, t} \leftrightarrow X_{n, t} \leftrightarrow X_{n, t-1} \leftrightarrow X_{n, t-2} \leftrightarrow \cdots$ and $Y_{n, t} \leftrightarrow X_{n, t} \leftrightarrow X_{n, t+1} \leftrightarrow X_{n, t+2} \leftrightarrow \cdots$ form two Markov chains for all $t=0,1,2,\ldots$. 
\end{assumption}
In other words, the safety level $Y_{n,t}$ depends on the status process \{$X_{n,t}$, $t= 0,1,2,...$\} only through the current status $X_{n,t}$. For instance, if the safety level is a function of the agent's status, expressed as $Y_{n,t} = g(X_{n,t})$ and the agent's status follows a Markov chain $X_{n, t} \leftrightarrow X_{n, t-1} \leftrightarrow X_{n, t-2} \leftrightarrow \cdots$, then the Assumption \ref{assumption1} holds \cite[Sec. 2.8]{IT_Cover}.

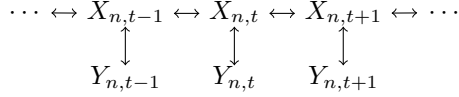
\begin{figure}
\centering
\begin{tikzpicture}
\matrix[matrix of math nodes,column sep=1em,row
sep=1.5em,cells={nodes={draw=none,minimum width=2em,inner sep=0pt}},
column 1/.style={nodes={draw=none}},
column 5/.style={nodes={draw=none}}] (m) {
\cdots & ~X_{n, t-1}~ & ~X_{n, t} ~& ~X_{n, t+1}~ & \cdots\\
& ~Y_{n, t-1}~ & ~Y_{n, t}~ & ~Y_{n, t+1}~ &\\
};
\foreach \X in {1, 2,3,4}
{\draw[<->] (m-1-\X) -- (m-1-\the\numexpr\X+1) ;
\ifnum\X=1,4
\else
\draw[<->] (m-1-\X) -- (m-2-\X) ;
\fi}
\end{tikzpicture}
\vspace{5pt}
\caption{Relationship between status $X_{n, t}$ and safety level $Y_{n, t}$.} 
\label{hidden_MC}
\vspace{-18pt}
\end{figure}

The agents transmit status update packets through $M$ unreliable wireless channels to the access point, where $N > M$. Let $\mu_n (t) \in \{0, 1\}$  denote the scheduling decision to  pull a packet from agent $n$ at time-slot $t$, defined as 
\begin{align} \label{decision}
\mu_n (t) \!\!=\!
\begin{cases}
& \!\!\!\!\!\!\!1, \!~\text{if status of agent $n$ is pulled in time-slot $t$},\\
& \!\!\!\!\!\!\!0, \!~\text{otherwise.}
\end{cases}
\end{align}
Upon receiving a pull request from the  access point, agent $n$ generates and transmits a time-stamped status packet $(X_{n,t},t)$ over one wireless channel. The transmission takes one time slot to reach the receiver. However, due to wireless channel fading, transmissions are subject to errors.
Let $p_n$ denote the probability of a successful transmission from agent $n$, irrespective of the selected wireless channel. 
Denote $\gamma_n (t) \in \{0, 1\}$ as the delivery indicator for agent $n$ at time-slot $t$, such that $\gamma_n (t)=1$ with probability $p_n$ if the transmission from agent $n$ is successful.
In other words, a packet will be received by the access point only if $\mu_n (t) =1$ and $\gamma_n (t) =1$. The delivery indicators $\gamma_n(t)$ are independent across both agents and time slots.

Due to channel sharing among the agents and transmission errors, the received information may not always be up to date. The most recent packet available at time $t$ is denoted as $X_{n, t-\Delta_n (t)}$, which was originally generated at time $t-\Delta_n (t)$. In remote estimation systems, the
time difference $\Delta_n (t)$ between the packet's generation time $t-\Delta_n (t)$ and the current time $t$ is known as \emph{age of information (AoI)} \cite{KaulYatesGruteser-Infocom2012, 2015ISITYates, Wiener_TIT, Ornee2021}. AoI quantifies the freshness of status updates received from agent $n$ and evolves according to:
\begin{align} \label{age}
\Delta_n (t+1) =
\begin{cases}
& \!\!\!\!\!\! 1, ~~~~~~~~~~~\text{if} {\thinspace} \mu_n (t) =1~{\thinspace} \text{and}~ {\thinspace} \gamma_n (t) =1, \\
& \!\!\!\!\!\! \Delta_n (t) +1, \text{otherwise.}
\end{cases}
\end{align}

\subsection{Safety Level Estimator}
At time slot $t$, estimator $n$ infers the safety level $Y_{n, t}$ using the up-to-date information at the receiver.
The risk due to unawareness of danger is quantified by  $L_n : \mathcal Y_n \times \mathcal Y_n \to \mathbb{R}$ for each agent $n$, where $L_n (y, \hat y)$ is the incurred loss if the actual safety level is $Y_{n,t} = y$ and the estimate is $\hat y$. The loss function $L_n (\cdot, \cdot)$ can be any function mapping from $\mathcal{Y}_n \times \mathcal{Y}_n$ to $\mathbb{R}$. 
An example of such a loss function is provided below.

\begin{example}\label{lossexample}
    Consider a scenario with three safety levels: {dangerous, cautious, safe}. If a dangerous situation is wrongly estimated as safe, the loss \emph{$L_n$(\emph{dangerous, safe})} will be significantly high due to the huge risk of damage. Conversely, if the safe situation is wrongly estimated as dangerous, the loss \emph{$L_n$(\emph{safe, dangerous})} will be small. This is because even if the estimation is incorrect, the risk of damage is small. Following this principle, the loss function is given by 
    \emph{$$L_n (\emph{dangerous, safe}) =1000,$$ 
    $$L_n (\emph{safe, dangerous}) = 5,$$
    $$L_n(\emph{cautious, safe}) =10,$$ 
    $$L_n(\emph{safe, cautious}) =1,$$
    $$L_n(\emph{cautious, dangerous}) =5,$$ 
    $$L_n(\emph{dangerous, cautious}) =100,$$}
    
    The loss for perfect estimation of the safety level is zero, expressed as, \emph{$L_n(\emph{safe, safe}) = L_n(\emph{cautious, cautious})= L_n(\emph{dangerous, dangerous}) =0.$}\footnote{In this example, if a dangerous situation is correctly detected, the loss is set as $L_n$(\emph{dangerous,dangerous})=0 due to perfect situational awareness. However, the loss function $L_n (\cdot, \cdot)$  can be freely designed according to the application. In other words, our paper is not limited to $L_n(\emph{dangerous, dangerous}) =0$ only.} One illustration of Example \ref{lossexample} is provided in Figure \ref{figure1}.
    \end{example} 

\begin{figure}
\centering
\includegraphics[width=7cm]{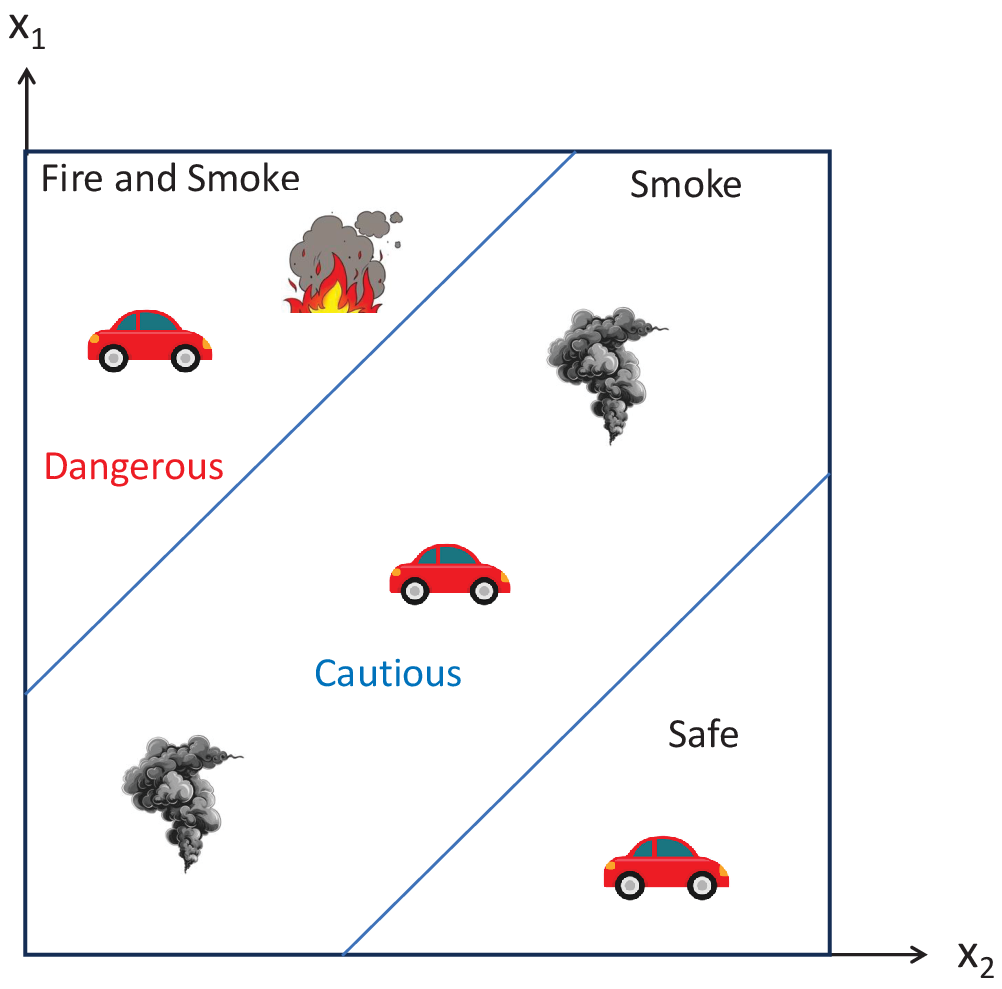} 
\vspace*{1.0mm} 
\caption{Safety regions with three possible situations \emph{dangerous, cautious,} and \emph{safe}, where the event of fire indicates a \emph{dangerous} situation, smoke indicates a \emph{cautious} situation, and \emph{safe}, otherwise.}\vspace{-0.0cm}
\label{figure1}
\end{figure} 

The loss values and the number of safety levels can be flexibly designed to fit different applications. By this, the loss function $L_n (y,\hat y)$ can be adjusted to meet specific situational awareness requirements, ensuring adaptability across various contexts. A key advantage of the  general loss function $L_n (\cdot, \cdot)$ is its ability to address safety risks arising from unawareness of potential danger—something traditional loss functions, such as $0$-$1$ loss, quadratic loss, and logarithmic loss, (See Appendix \ref{loss_definitions} for definitions) fail to capture. 



The information available to estimator $n$ at time $t$ is given by
\begin{align} \label{history}
\mathbf{H}_{n, t} = (X_{n, \tau-\Delta_n (\tau)}, \Delta_n (\tau), \mu_n (\tau-1), \gamma_n (\tau -1))_{\tau=0}^{t}
 ,
\end{align}
which includes all historically received packets, AoI values, scheduling decisions, and delivery indicators for agent $n$ up to time slot $t$.
}

\subsection{Jointly Optimal Scheduler and Estimator Design Problem} \label{joint_design}
We denote the estimator $n$ by a function $\phi_n (\mathbf{H}_{n, t}) \in \mathcal{Y}_n$ that maps any input from $\mathcal{H}_{n, t}$ to $\mathcal{Y}_n$. We also denote the set $\Phi_n$ as the set of all possible estimator functions $\phi_n (\cdot)$ that map any input from $\mathcal{H}_{n, t}$ to $\mathcal{Y}_n$.



Next, let $\pi = (\mu_n(0), \mu_n(1), \ldots)_{n=1}^N$
represent a scheduling policy, where 
$\mu_n (t)$ is defined in \eqref{decision}.
Let $\Pi$ denote the set of all  causal scheduling policies in which the scheduling decision $(\mu_n(t))_{n=1}^N$ at time slot $t$ is made based on the historical information $(\mathbf{H}_{n,t})_{n=1}^N$, where $\mathbf{H}_{n, t}$ is defined in \eqref{history}.  


Our goal is to find a scheduling policy $\pi$ and estimators $\phi_1, \phi_2, \ldots, \phi_N$ that minimize the time-average sum of the expected loss across all $N$ agents. The joint optimization problem is formulated as 
\begin{align}
\!\mathsf{L}_{\text{opt}}\! &=\! \inf_{\pi \in \Pi, \phi \in \Phi}\! \limsup_{T \to \infty}\! \sum_{n=1}^{N} \! \frac{1}{T} \!\sum_{t=0}^{T-1} \!\mathbb{E} \big[L_n(Y_{n, t}, \phi_n (\mathbf{H}_{n, t}))\big] \label{problem} \\
& \quad \quad \text{s.t.} \sum_{n=1}^{N} \mu_n (t) \leq M, \mu_n (t) \in \{0,1\}, t = 0,1, \ldots,\label{constraint}
\end{align}
where $\phi=(\phi_1, \phi_2, \ldots, \phi_n)$ is an $N$-tuple in which the $n$-th element $\phi_n$ is the $n$-th estimator function, $\Phi=\Phi_1 \times \Phi_2 \times \ldots \times \Phi_n$ is the $N$-fold cartesian product of all the tuples $\phi$, and $\mathsf{L}_{\text{opt}}$ is the optimum objective value of the problem \eqref{problem}-\eqref{constraint}. Since the system operates with $M$ available channels, the constraint $\sum_{n=1}^{N} \mu_n (t) \leq M$ ensures that at most $M$ agents can transmit in any given time slot.

\section{Problem Simplification}\label{probelm_simplification}


Problem \eqref{problem}-\eqref{constraint} is significantly challenging because  the available information $(\mathbf{H}_{n, t})_{n=1}^N$ is expanding at every time $t$. Using sufficient statistics, the problem \eqref{problem}-\eqref{constraint} can be greatly simplified, as outlined in the following theorem:

\begin{theorem}\label{thm_sufficient_statistic}
If Assumption \ref{assumption1} holds, then
the following assertions are true:

(i) $(\Delta_n (t), X_{n, t - \Delta_n (t)})$ 
is a sufficient statistic of $\mathbf{H}_{n, t}$ for estimating $Y_{n, t}$ at time $t$,

(ii) $(\Delta_n (t), X_{n, t - \Delta_n (t)})_{n=1}^N$ 
is a sufficient statistic of $(\mathbf{H}_{n, t})_{n=1}^N$ for making the scheduling decision $(\mu_n (t))_{n=1}^N$ at time $t$ in \eqref{problem}-\eqref{constraint}.

\end{theorem}

\begin{proof}
See Appendix \ref{ptheorem1}.
\end{proof}


Theorem\ref{thm_sufficient_statistic}(i)-(ii) shows that $(\Delta_n (t), X_{n, t - \Delta_n (t)})$ is a sufficient statistics for finding both the optimal estimator and the optimal scheduler. Using the sufficient statistics, Lemma \ref{lemmaoptimalestimator} determines optimal estimation for any scheduling policy $\pi$.

\begin{lemma}\label{lemmaoptimalestimator}
   Under any scheduling policy $\pi \in \Pi$, if  $\Delta_n(t)=\delta$ and $X_{n, t-\Delta_n (t)}=x$, then the output $y_n$ of the $n$-th optimal estimator at time $t$ minimizes the following optimization problem:
\begin{align}\label{optimalestimator}
   \inf_{y_n \in \mathcal Y_n} \mathbb{E}_{Y \sim P_{Y_{n, t}|\Delta_n (t) = \delta, X_{n, t-\delta}=x}} \big[\!L_n(Y, y_n\!)\!\big].
\end{align}
\end{lemma}

\begin{proof}
See Appendix \ref{proof_lemma_joint}.
\end{proof}

By using Lemma \ref{lemmaoptimalestimator}, we construct an estimator function $f_n(\delta, x)$ that maps any $(\delta, x) \in \mathbb Z^{+} \times \mathcal X_n$ to $\mathcal Y_n$ as follows:
\begin{align}\label{est_func}
    f_n(\delta, x)=\argmin_{y_n \in \mathcal Y_n} \mathbb{E}_{Y \sim P_{Y_{n, t}|\Delta_n (t) = \delta, X_{n, t-\delta}=x}} \big[\!L_n(Y, y_n\!)\!\big]. 
\end{align}

Using Theorem \ref{thm_sufficient_statistic} and the estimator function in \eqref{est_func}, 
we obtain the following equivalent optimization problem. Given the estimator function $f_n$ defined in \eqref{est_func} for all $n=1,2, \ldots, N$, the problem \eqref{problem}-\eqref{constraint} is equivalent to 
\begin{align}
\!\mathsf{L}_{\text{opt}}\! =\!\! & \inf_{\pi \in \Pi_e}\! \limsup_{T \to \infty}\!\! \sum_{n=1}^{N} \!\! \frac{1}{T} \!\!\sum_{t=0}^{T-1} \!\mathbb{E} \!\big[\!L_n(Y_{n, t}, f_n (\Delta_n (t), X_{n, t-\Delta_n (t)})\!)\!\big] \label{problem_eq} \\
& ~\text{s.t.} \sum_{n=1}^{N} \mu_n (t) \leq M, \mu_n (t) \in \{0,1\}, t = 0,1, \ldots, \label{constraint_eq}
\end{align} 
where $\Pi_e \in \Pi$ is the set of all causal scheduling policies in which the scheduling decision
$(\mu_n(t))_{n=1}^{N}$ at time slot $t$ is made by using all the latest received observations $(X_{n, t-\Delta_n (t)})_{n=1}^{N}$ and their AoI values $(\Delta_n (t))_{n=1}^{N}$.
\!\!

By this, we obtain a new problem \eqref{problem_eq}-\eqref{constraint_eq} with a reduced state space, where the state $(\Delta_n (t), X_{n, t-\Delta_n(t)})$ of the $n$-th agent includes both the latest received observation $X_{n, t-\Delta_n (t)}$ of agent $n$ and its AoI $\Delta_n (t)$ at time slot $t$.


\begin{remark}
One major contribution of our study is the reduction of the state space using the sufficient statistic of the history (i.e., the latest received observation $X_{n , t-\Delta_n (t)}$ and its AoI $\Delta_n (t)$) while  existing studies utilized belief MDP or POMDP formulation. The belief states used in \cite{Liu_TIT, Ouyang_TMC, Ouyang_Infocom, Javidi_TIT, Ansell_2003, Neely_2010, chen2021scheduling, chen2022index} help  reduce the state space. However, the state space is still uncountable \cite{Liu_TIT, Javidi_TIT, Ansell_2003, Neely_2010}. Although some attempts have been made to make the state space countable under a positive recurrent assumption and using a sufficiently large truncated AoI value, the state space still exhibits a quadratic increase with the AoI \cite{chen2021scheduling, chen2022index, Ouyang_Infocom, Ouyang_TMC}. 
Specifically, for a truncated set $\{1, 2, \ldots, \tau\}$ of AoI values, the state space increases as $\tau \times |\mathcal X_n|^2$, where $X_{n, t} \in \mathcal X_n$ represents the $n$-th bandit process. The difference between the formulation in \cite{chen2021scheduling, chen2022index, Ouyang_Infocom, Ouyang_TMC} and problem \eqref{problem_eq}-\eqref{constraint_eq} is that we do not need to utilize belief states. By utilizing the sufficient statistic of the history observation, i.e., the latest received observation and the corresponding AoI value, we obtain a significantly smaller state space which demonstrate linear growth with AoI, such as $\tau \times |\mathcal X_n|$.
\end{remark}

The problem \eqref{problem_eq}-\eqref{constraint_eq} characterizes the fundamental performance limits in maximizing situational awareness through pull-based remote estimation. Specifically, the objective function of Problem \eqref{problem_eq}-\eqref{constraint_eq} can be interpreted using information-theoretic metrics. To this end, we consider the generalized entropy  $H_L(Y)$ of a random variable  $Y$ associated with a loss function  $L: \mathcal Y \times \mathcal Y \rightarrow R$, defined as \cite{dawid1998coherent, farnia2016minimax, Shisher2022}.
\begin{align} \label{gen_def_L_entropy}
H_L (Y) = \min _{y \in \mathcal{Y}} \mathbb{E}_{Y \sim P_Y} [L (Y, y)].
\end{align}
Then, the $L$-conditional entropy of $Y$ given $X = x$ and the $L$-conditional entropy of $Y$ given $X$ can be defined as  \cite{dawid1998coherent, farnia2016minimax, Shisher2022}
\begin{align} \label{L_cond_en}
H_L (Y | X =x)
=& \min _{y \in \mathcal{Y}} \mathbb{E}_{Y \sim P_{Y|X=x}} [L (Y, y)],\\  
H_L (Y | X)
=& \sum_{x \in \mathcal{X}} P_X (x) H_L (Y | X =x). \label{L_cond_en_1}
\end{align}
Equations \eqref{L_cond_en} and \eqref{L_cond_en_1} along with the Data Processing Inequality \cite[Sec. 2.8]{IT_Cover} will be utilized in Section \ref{discussions} to get a nice property of the scheduling policy.

Building upon the insights of \cite{shisher2021age, Shisher_2024, shisher2023learning}, we utilized $L$-conditional entropy $H_L(Y_t|X_{t-\Delta(t)}=x, \Delta(t)=\delta)$ given both the AoI $\delta$ and the knowledge of the received observation $x$ to measure the impact of the AoI and the information content in remote estimation and prediction. 
Compared to \cite{shisher2021age, Shisher_2024, shisher2023learning}, we consider a signal-aware scheduling scheme (the decision-maker utilizes the knowledge of the signal value) with the goal of minimizing the performance loss caused by situational unawareness where \cite{shisher2021age, Shisher_2024, shisher2023learning} focused on signal-agnostic scenario (the decision-maker does not utilize signal value).
\begin{lemma} \label{lemma_entropy}
It holds that
\begin{align} \label{equality}
H_L (Y_{n, t} |\Delta_n(t), X_{n, t-\Delta_n(t)})=\sum_{(\delta, x) \in \mathbb Z \times \mathcal X_n} P_{\Delta_n(t),X_{n, t-\delta}}(\delta,x)\inf_{f_n(\delta, x) \in \mathcal Y_n} \mathbb{E}_{Y \sim P_{Y_{n, t}|\Delta_n (t) = \delta, X_{n, t-\delta}=x}} \big[L_n(Y, f_n (\delta, x))\big],
\end{align}
where $H_L (Y_{n, t} |\Delta_n(t), X_{n, t-\Delta_n(t)})$ is the generalized conditional entropy of $Y_{n,t}$ given the latest received observation $X_{n, t-\Delta_n(t)}$ at time slot $t$ and it's AoI value $\Delta_n(t)$.
\end{lemma}

\begin{proof}
See Appendix \ref{proof_lemma_entropy}.
\end{proof}

The probability $P_{\Delta_n(t),X_{n, t-\delta}}(\delta,x)$ in Lemma \ref{lemma_entropy} depends on the scheduling policy $\pi$.
By using Lemma \ref{lemma_entropy}, problem \eqref{problem_eq}-\eqref{constraint_eq} can be written as
\begin{align}
\mathsf{L}_{\text{opt}}= 
& \inf_{\pi \in \Pi} \limsup_{T \to \infty} \sum_{n=1}^{N} \frac{1}{T} \sum_{t=0}^{T-1} H_L (Y_{n, t}|\Delta_n (t), X_{n, t-\Delta_n (t)})  \label{problem_eq_1} \\
& ~\text{s.t.} \sum_{n=1}^{N} \mu_n (t) \leq M, \mu_n (t) \in \{0,1\}, t = 0,1, \ldots.\label{constraint_eq_1}
\end{align}

\section{Scheduling Policy Design for Solving \eqref{problem_eq_1}-\eqref{constraint_eq_1}}


Problem \eqref{problem_eq_1}-\eqref{constraint_eq_1} is a Restless Multi-armed Bandit (RMAB). The problem \eqref{problem_eq_1}-\eqref{constraint_eq_1} is called “restless” because the state and the penalty of each bandit $n$ continues to evolve over time, regardless of whether the agent $n$ is selected for transmission \cite{whittle_restless}. Because we are able to simplify the original problem \eqref{problem}-\eqref{constraint}, the RMAB \eqref{problem_eq_1}-\eqref{constraint_eq_1} has a reduced state space compared to \eqref{problem}-\eqref{constraint}. However,  even with the reduced state space, solving RMAB problems and finding optimal solutions are significantly challenging.
A Whittle index policy is known to be an efficient approach to solving RMAB problems which requires to satisfy a condition called indexability\cite{whittle_restless}, \cite{weber1990index}. A key challenge in solving problem \eqref{problem_eq_1}-\eqref{constraint_eq_1} is that indexability is very difficult to establish. This difficulty arises due to the following reasons: (i) The state of each bandit of RMAB \eqref{problem_eq_1}-\eqref{constraint_eq_1} exhibits a complicated transition, (ii) the transmission channels are unreliable, and (iii) the expected penalty associated with each bandit can be non-monotonic function of the AoI while most of the previous studies considered monotonic penalty functions of AoI \cite{kosta2017age, sun2019sampling, sun2017update, tripathi2019whittle}.
Hence, \eqref{problem_eq_1}-\eqref{constraint_eq_1} is a more challenging problem than the problems studied in \cite{kosta2017age, sun2019sampling, sun2017update, tripathi2019whittle, Chen_2022, chen2022index}. 
However, we are able to develop a Maximum Gain First (MGF) policy that does not need to satisfy indexability. 

We solve problem \eqref{problem_eq_1}-\eqref{constraint_eq_1} in three-steps: (i) We first relax constraint \eqref{constraint_eq_1} and utilize Lagrangian dual decomposition to decompose the original problem into separated per-bandit problems; (ii) Next, we develop a Maximum Gain First (MGF) policy that does not need to satisfy any indexability condition; (iii) Finally, we prove that the developed policy is asymptotically optimal.


\subsection{Relaxation and Lagrangian Decomposition}
In standard RMAB problems, the constraint \eqref{constraint_eq_1} needs to be satisfied with equality, i.e., exactly $M$ bandits are activated at any time slot $t$. However, in our problem, constraint \eqref{constraint_eq_1} activates less than $M$ bandits at any time $t$. Following \cite[Section 5.1.1]{verloop2016asymptotically}, \cite[Section IV-A]{ornee2023whittle}, we introduce $M$ additional \emph{dummy bandits} that never change state and therefore, they incur no penalty. If a \emph{dummy bandit} is activated, it occupies one channel but does not incur any penalty. Let $\mu_0 (t) \in \{1, 2, \ldots, M\}$ be the number of \emph{dummy bandits} that are activated at time slot $t$. These dummy bandits are introduced to establish the asymptotic optimality discussed in Section \ref{asymptotic optimality}. After incorporating these \emph{dummy bandits}, the RMAB \eqref{problem_eq_1}-\eqref{constraint_eq_1} can be expressed as
\begin{align}
\mathsf{L}_{\text{opt}}\! =
& \inf_{\pi \in \Pi}\! \limsup_{T \to \infty}\sum_{n=1}^{N} \frac{1}{T} \sum_{t=0}^{T-1} H_L (Y_{n, t}|\Delta_n (t), X_{n, t-\Delta_n (t)}) \label{problem_dummy} \\
& ~\text{s.t.} \sum_{n=0}^{N} \mu_n (t) = M, \label{constraint_dummy}\\
& ~~~~~\mu_0 (t)\!\! \in \!\{1, 2, \ldots, M\}, t = 0,1, \ldots,\\
& ~~~~~\mu_n (t) \in \{0, 1\}, n = 1,2, \ldots, t = {0,1,\ldots}, \label{constraint_dummy_2}
\end{align}
which is an RMAB with an equality constraint. 
Because the \emph{dummy bandits} never change state, problem \eqref{problem_eq_1}-\eqref{constraint_eq_1} and \eqref{problem_dummy}-\eqref{constraint_dummy_2} are equivalent. 

Next, we follow the standard relaxation and Lagrangian decomposition procedure for RMAB \cite{whittle_restless} and relax the constraint \eqref{constraint_dummy} and obtain the following relaxed problem: 
\begin{align} 
\mathsf{L}_{\text{opt}} =
& \inf_{\pi \in \Pi} \!\limsup_{T \to \infty} \!\! \sum_{n=1}^{N} \frac{1}{T} \sum_{t=0}^{T-1} H_L (Y_{n, t}|\Delta_n (t), X_{n, t-\Delta_n (t)}), \label{relaxed_problem}  \\
& ~\text{s.t.} \limsup_{T \to \infty}  \sum_{n=0}^{N} \mathbb{E} \bigg[\frac{1}{T} \sum_{t=0}^{T-1} \mu_n (t) \bigg] = M, \label{relaxed_cons_2}\\
& ~~~~~\mu_0 (t) \in \{1,2,\ldots, M\}, t = {0,1,\ldots}, \\
& ~~~~~\mu_n (t) \in \{0, 1\}, n = 1,2, \ldots, t = {0,1,\ldots}. \label{relaxed_cons}
\end{align}
The relaxed constraint \eqref{relaxed_cons_2} needs to be satisfied on average, instead of satisfying at every time slot $t$. To solve the relaxed problem \eqref{relaxed_problem}-\eqref{relaxed_cons}, we take a dual cost $\lambda \geq 0$ (also known as Lagrange multiplier) for the relaxed constraint. The dual problem is given by
\begin{align} \label{dual_problem}
\sup_{\lambda \geq 0} \bar L (\lambda), 
\end{align}
where 
\begin{align} \label{decomposed}
\bar L (\lambda) =
 \!\inf_{\pi \in \Pi} \!\limsup_{T \to \infty} \bigg[ \sum_{n=1}^{N} \!\frac{1}{T} \!\!\sum_{t=0}^{T-1}\!  H_L (Y_{n, t}|\Delta_n (t), X_{n, t-\Delta_n (t)}) + \lambda \bigg(\mathbb{E} \bigg[\sum_{n=0}^{N} \mu_n (t)\bigg] - M\bigg) \bigg].
\end{align}
The term $\frac{1}{T} \sum_{t=0}^{T-1} \sum_{n=0}^{N} \lambda M$ in \eqref{decomposed} does not depend on policy $\pi$ and hence can be removed. For a given $\lambda$, problem \eqref{decomposed} can be decomposed into $(N+1)$ separated sub-problems and each sub-problem associated with agent $n$ is formulated as
\begin{align} \label{per_arm_problem}
\bar L_n (\lambda) =
\inf_{\pi_n \in \Pi_n} \!\!\limsup_{T \to \infty} \bigg[\frac{1}{T} \sum_{t=0}^{T-1} H_L (Y_{n, t}|\Delta_n (t), X_{n, t-\Delta_n (t)}) + \lambda \mathbb{E} [\mu_n (t)] \bigg],
\end{align}
where $\bar L_n (\lambda)$ is the optimum value of \eqref{per_arm_problem}, $\pi_n= (\mu_n (0), 
\mu_n (1), \ldots)$ is the sub-scheduling policy for agent $n$, and $\Pi_n$ is the set of all causal sub-scheduling policies of agent $n$. Problem \eqref{per_arm_problem} is a per-bandit problem associated with bandit $n$. On the other hand, the sub-problem associated with the \emph{dummy bandits} is given by
\begin{align}
\bar L_0 (\lambda) = \inf_{\pi_0 \in \Pi_0} \limsup_{T \to \infty} \mathbb{E} \bigg[\frac{1}{T} \sum_{t=0}^{T-1} \lambda \mu_0 (t)\bigg], \label{per_arm_problem_dummy}
\end{align}
where $\bar L_0 (\lambda)$ is the optimum value of \eqref{per_arm_problem_dummy}, $\pi_0= \{\mu_0 (t), t = 0,1,\ldots\}$, and $\Pi_0$ is the set of all causal activation policies $\pi_0$.

{ \subsection{Maximum Gain First (MGF) Policy}


For a given transmission cost $\lambda$, the per-bandit problem \eqref{per_arm_problem} can be cast as an average-cost infinite horizon 
MDP with state $(\Delta_n (t), X_{n, t-\Delta_n (t)})$. The state associated with each bandit $n$ is the latest received observation $X_{n, t-\Delta_n (t)}$ and its AoI $\Delta_n (t)$ at time slot $t$. The action is defined by $\mu_n (t) \in \{0, 1\}$ which denotes the scheduling decision for agent $n$ at every time slot $t$. Each MDP associated with each bandit $n$ has two actions: active and passive. 
If a packet from agent $n$ is requested and submitted to a channel at time slot $t$, then restless bandit $n$ takes an active action at time slot $t$; otherwise, bandit $n$ is made passive at time slot $t$.

If bandit $n$ is not scheduled for transmission (i.e., $\mu_n (t) = 0$), then the AoI will increase by 1, i.e., $\Delta_n (t) = \Delta_n (t-1) +1$ and the observation $X_{n, t-\Delta_n (t)}$ will remain in the same state. If bandit $n$ is scheduled for transmission (i.e., $\mu_n (t) =1$) and the transmission succeeds with probability $p_n$, then the AoI will drop to 1, i.e., $\Delta_n (t) = 1$ and the observation will change to a new received value $X_{n, t-1}$; otherwise, if transmission fails with probability $1-p_n$, then the AoI will increase by 1, i.e., $\Delta_n (t) = \Delta_n (t-1) +1$ and the observation will remain the same. 

We solve \eqref{per_arm_problem} by using dynamic programming \cite{bertsekas2011dynamic}.} 
The Bellman optimality equation for the MDP in \eqref{per_arm_problem} for given $\Delta_n (t) = \delta$ and $X_{n, t-\Delta_n (t)} =x$ is 
\begin{align}
h_{n, \lambda} (\delta, x) = \min_{\mu \in \{0,1\}} Q_{n, \lambda} (\delta, x, \mu),
\end{align}
where $h_{n, \lambda} (\delta, x)$ is the relative-value function of the average-cost MDP and $Q_{n, \lambda} (\delta, x, \mu)$ is the relative action-value function defined as
\begin{align} \label{action_function}
Q_{n, \lambda} (\delta, x, \mu) =
\begin{cases}
q_n (\delta, x) - \bar L_n (\lambda) + h_{n, \lambda} (\delta+1, x) , & \text{if} {\thinspace} \mu =0, \\
q_n (\delta, x) - \bar L_n (\lambda) + (1-p_n) h_{n, \lambda} (\delta+1, x)
+ p_n  \mathbb{E} [h_{n, \lambda} (1, X_{n, 0}) | X_{n, -\delta} = x]+ \lambda,  
& \text{otherwise},
\end{cases}
\end{align}
where $q_n (\delta, x)$ is given by
\begin{align} \label{loss_ex}
q_n (\delta, x) = H_L (Y_{n, t} |\Delta_n (t) = \delta, X_{n, t-\delta} =x),   
\end{align}
and \eqref{action_function} holds because $X_{n, t}$ is a time-homogeneous Markov chain.
The relative value function $h_{n, \lambda} (\delta, x)$ can be computed by using relative value iteration algorithm for average-cost MDP \cite{bertsekas2011dynamic}. 

{ Let $\pi_{n, \lambda} ^{*} = (\mu_{n, \lambda} ^{*} (1), \mu_{n, \lambda} ^{*} (2), \ldots)$ be an optimal solution to \eqref{per_arm_problem}. The optimal decision at time slot $t$ for bandit $n$ is given by
\begin{align} \label{action}
\mu_{n, \lambda} ^{*} (t) = \argmin_{\mu \in \{0,1\}} Q_{n, \lambda} (\Delta_n (t), X_{n, t-\Delta_n (t)}, \mu),
\end{align}
where the dual cost is iteratively updated using the stochastic dual sub-gradient ascent method with step size $\beta > 0$ \cite{nedic2008subgradient}:
\begin{align} \label{lambda_update}
\lambda (j+1) = \lambda (j) + \frac{\beta}{j} \bigg(\frac{1}{T} \sum_{n=0}^{N} \sum_{t=0}^{T-1} \mu^*_{n, \lambda (j)} (t) - M\bigg),
\end{align}
for $j$-th iteration.
Let $\lambda^ *$ be the optimal dual cost to problem \eqref{dual_problem} to which $\lambda (t)$ converges. Then, we can apply $(\pi_{n, \lambda^{*}})_{n=0}^{N}$ for the relaxed problem \eqref{relaxed_problem}-\eqref{relaxed_cons}. But applying this policy to the  problem \eqref{problem_dummy}-\eqref{constraint_dummy} may violate the constraint \eqref{constraint_dummy}. 

\begin{algorithm}[t] 
\caption{Maximum Gain First Policy for Solving \eqref{problem_eq_1}-\eqref{constraint_eq_1}} \label{algo1}
\begin{algorithmic}[1]
\State At time $t=0$:
\State Input $\alpha_{n} (\delta, x)$ in \eqref{net_gain} for every bandit $n=1,2,\ldots,N$. 
\State For all time $t = 0, 1, \ldots$,
\State Update $(\Delta_n (t), X_{n, t-\Delta_n (t)})$ for all bandits $n=1,2,\ldots, N$.
\State Update current ``gain" $\alpha_{n} (\Delta_n (t), X_{n, t-\Delta_n (t)})$ for all bandits $n=1,2,\ldots, N$.
\State Choose at most $M$ bandits with the highest positive ``gain". 
\end{algorithmic}
\end{algorithm}

\begin{definition}[\textbf{Gain Index}]\label{gainindex}
Following \cite{chen2022index, shisher2023learning}, we define the ``gain" $\alpha_{n} (\delta, x)$ as
\begin{align} \label{net_gain}
\alpha_{n} (\delta, x) = Q_{n, \lambda^*} (\delta, x, 0) - Q_{n, \lambda^*} (\delta, x, 1).
\end{align}
If $Q_{n, \lambda^*} (\delta, x, 0) > Q_{n, \lambda^*} (\delta, x, 1)$, i.e., $\alpha_{n} (\delta, x) > 0$, it is optimal to schedule bandit $n$. If $Q_{n, \lambda^*} (\delta, x, 0) < Q_{n, \lambda^*} (\delta, x, 1)$, i.e., $\alpha_{n} (\delta, x) < 0$, it is optimal to not to schedule bandit $n$.
\end{definition}

The Algorithm for solving RMAB \eqref{problem_eq_1}-\eqref{constraint_eq_1} is provided in Algorithm \ref{algo1} which activates the at most $M$ bandits with the highest positive ``gain" index at any time slot $t$. We prove the asymptotic optimality of the Maximum Gain First (MGF) policy in Section \ref{asymptotic optimality}.

\begin{remark}
Our method of developing the MGF policy in Algorithm \ref{algo1} by optimizing the average conditional entropy in \eqref{problem_eq_1}-\eqref{constraint_eq_1} is conceptually similar to the "water-filling" problem. In "water-filling", power resources are allocated to channels with lower noise and higher Signal-to-noise Ratio (SNR). Similarly, in this study, we obtain the "gain" for given $\Delta_n (t) = \delta$ and $X_{n, t-\Delta_n (t)} =x$ defined in Definition \ref{def2}. Then, we schedule agent $n$ with the highest gain, which effectively reduces the conditional entropy.
\end{remark}


\begin{remark}
Further discussion on a special case of single-source, single-channel (i.e., $N=M=1$) is provided in Appendix \ref{single_source}. For this special case, we find that for single-source, single-channel, it is always better to send. Earlier studies \cite{Sun_TIT_2017, sun2019sampling, Shisher_2024} reported that for single-source, single-channel scenario, it is not always better to take the active action for the following two cases: (i) if the penalty function is a non-monotonic function of the age, (ii) if the transmission times are random. However, in our study, though the penalty $q_1 (\delta, x)$ is not necessarily a monotonic function of the age, we show that it is always better to send. 
\end{remark}

\subsection{Asymptotic Optimality} \label{asymptotic optimality}
In this section, we demonstrate that the ``MGF Policy" in Algorithm \ref{algo1} is asymptotically optimal as the number of agents increases while maintaining a fixed ratio between the number of agents and the number of channels. Let $\pi_{\text{gain}}$ denote the policy presented in Algorithm \ref{algo1}. Following the standard practice, we consider a set of bandits to be in the same class if they share identical penalty functions and transition probabilities. 

\begin{definition}[\textbf{Asymptotic optimality}]\label{defAsymptotically optimal}
Consider a ``base'' system with $N$ bandits, $M$ channels, and $M$ dummy bandits. Let $\mathsf{L}^r_{\text{gain}}$ represent the long-term average cost under policy $\pi_{\text{gain}}$ for a system that includes $rN$ bandits, $rM$ channels, and $rM$ dummy bandits with $N+1$ class of bandits including one dummy bandit class with $r \in \mathbb{Z}^{+}$. The policy $\pi_{\text{gain}}$ is asymptotically optimal if $\mathsf{L}^r_{\text{gain}}=\mathsf{L}^r_{\text{opt}}$ as the number of bandits in each class increases by $r$ times.  
\end{definition}

We denote by $V^n_{\delta, x}(t)$ be the fraction of class $n$ bandits in state $(\delta, x)$ at time $t$. Then, we define

\begin{align}
 v^n_{\delta, x}=\limsup_{T\rightarrow \infty}\sum_{t=0}^{T-1} \frac{1}{T} \mathbb E [ V^n_{\delta, x}(t)].
\end{align}
We further use the vectors $\mathbf{V}^n(t)$ and $\mathbf{v}^n$ to contain $V^n_{\delta, x}(t)$ and $v^n_{\delta, x}$, respectively for all $\delta=1, 2, \ldots, \delta_{\mathrm{bound}}$ and $x \in \mathcal X$. Truncated AoI space will have little to no impact if $\delta_{\mathrm{bound}}$ is very large. This is because $L_n(\delta, x)$ for any $x \in \mathcal X$ achieves a stationary point as the AoI $\delta$ increases to a large value.  

For a policy $\pi$, we can have the following mapping 
\begin{align}
\Psi_{\pi} ((\mathbf{v}^{n})_{n=1}^N)=\mathbb{E} [(\mathbf{V}^n (t+1))_{n=1}^N | (\mathbf{V}^n (t))_{n=1}^N  = (\mathbf{v}^n)_{n=1}^N].
\end{align}
We define the $t$-th iteration of the maps $\Psi_{\pi, t \geq 0} (\cdot)$ as follows
\begin{align}
& \Psi_{\pi, 0} ((\mathbf{v}^n)_{n=1}^N) = (\mathbf{v}^n)_{n=1}^N, \\
& \Psi_{\pi, t+1} ((\mathbf{v}^n)_{n=1}^N) = \Psi_{\pi}(\Psi_{\pi, t}((\mathbf{v}^n)_{n=1}^N)).
\end{align}

Now, we are ready to define a global attractor condition. 

\begin{definition} \label{def2}
\textbf{Uniform Global attractor.} An equilibrium point $({\mathbf v^{n}_{\text{opt}}})_{n=1}^{N}$ given by the optimal solution of \eqref{problem_eq_1}-\eqref{constraint_eq_1} is a uniform global attractor of $\Psi_{\pi, t \geq 0} (\cdot)$, i.e., for all $\epsilon > 0$, there exists $T(\epsilon)$ such that for all $t \geq T(\epsilon)$ and for all $({\mathbf v^{n}_{\text{opt}}})_{n=1}^{N})$, one has $||\Psi_{\pi, t} (({\mathbf v^{n}})_{n=1}^{N})-({\mathbf v^{n}_{\text{opt}}})_{n=1}^{N}))||_1 \leq \epsilon$.
\end{definition}

\begin{theorem} \label{optimality}
Under the uniform global attractor condition in Definition \ref{def2}, the policy \emph{$\pi_{\text{gain}}$} is asymptotically optimal. 
\end{theorem}

\begin{proof}
See Appendix \ref{proof_optimality}.
\end{proof}

Unlike the Whittle Index policy, we do not need to establish any indexability condition in the MGF policy. However, this is still asymptotically optimal.


\begin{figure*}[t]
\vspace{-0.3cm}
\centering
\includegraphics[width=12cm]{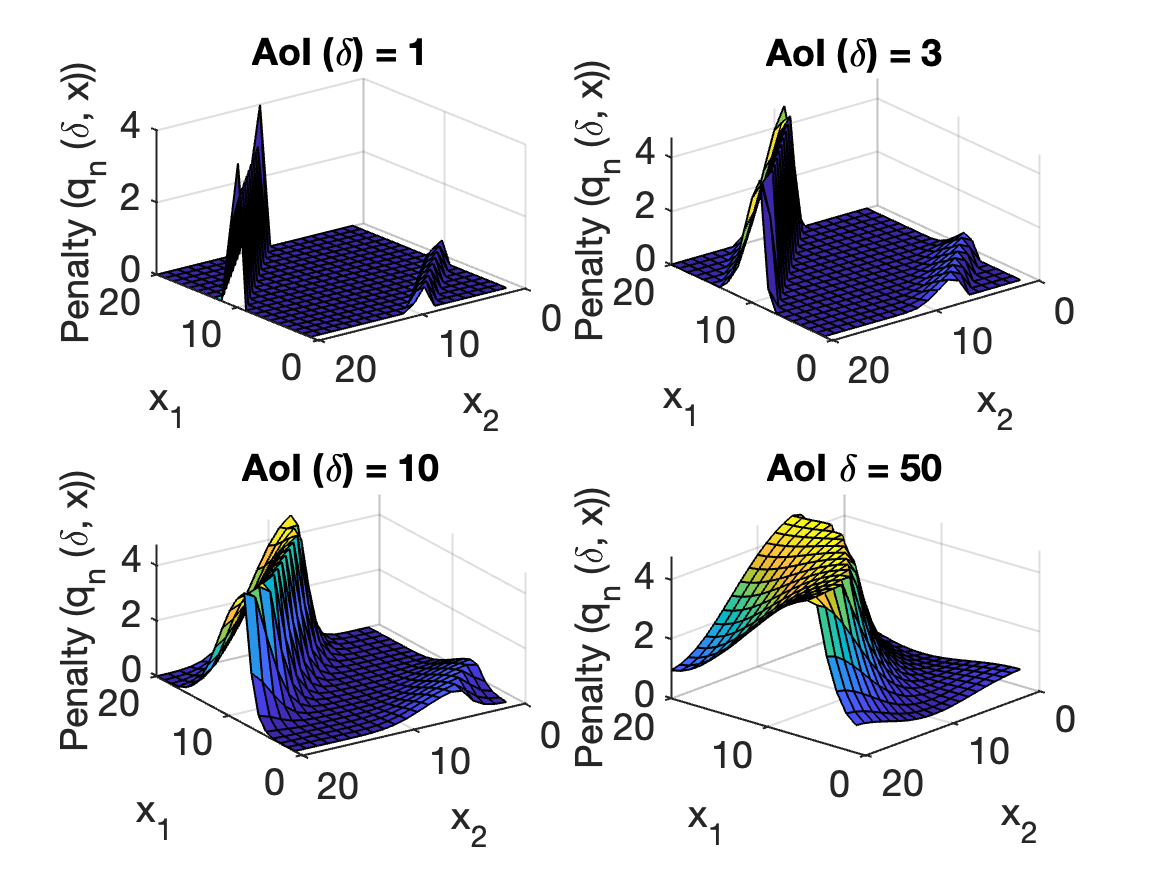}  
\vspace*{2.0mm} 
\caption{Penalty $q_n (\delta, x)$ vs received observation $x$ for fixed AoI $\delta$.}\vspace{-0.0cm}
\label{figure3}
\end{figure*}

\section{Discussions} \label{discussions}

The solution of problem \eqref{problem}-\eqref{constraint} yields an interesting relationship between an agent's proximity to safety boundaries and the required update frequency:
{\bf frequent update near the safety boundaries is optimal}. Specifically, analyzing the penalty function $q_n (\delta, x)$ for a given $\delta$ demonstrates that at the boundary of each safety region, the scheduler should update frequently and at far from the boundary, the update can be less frequent.
The details are provided below:

}

%

{ 

To understand this characteristic, we do the following experiment: consider a safety-critical system where $N$ agents (e.g., cars) are moving in a region illustrated in Figure \ref{figure1}. This region is equally divided into 400 positions and the received observation $X_{n, t}$ of agent $n$ is represented by the position $x_n = (x_{n,1}, x_{n,2})$ at time $t$. The safety level $Y_{n, t}$ is divided into three regions: \emph{\{safe, cautious, dangerous\}}. An agent $n$ can randomly move in any of the four directions: \emph{up, down, left}, and \emph{right} with equal probability 0.2. If agent $n$ is in the leftmost position, then moving left means it will stay in the same position, similar criteria are applied for the rightmost, upmost, and downmost positions. The losses considered in this experiment are the same as in Example \ref{lossexample}: $ L$\emph{(dangerous, safe)} $=1000, L$\emph{(safe, dangerous)} $=5,
L$\emph{(cautious, safe)} $= 10, L$\emph{(safe, cautious)} $= 1, 1,L$\emph{(dangerous, cautious)} $=100$, $L$\emph{(cautious, dangerous)}$= 5$, and $L$\emph{(dangerous, dangerous)} $=L$\emph{(cautious, cautious)} $=L$\emph{(safe, safe)} $=0$.

\begin{figure}
\vspace{-0.3cm}
\centering
\includegraphics[width=9cm]{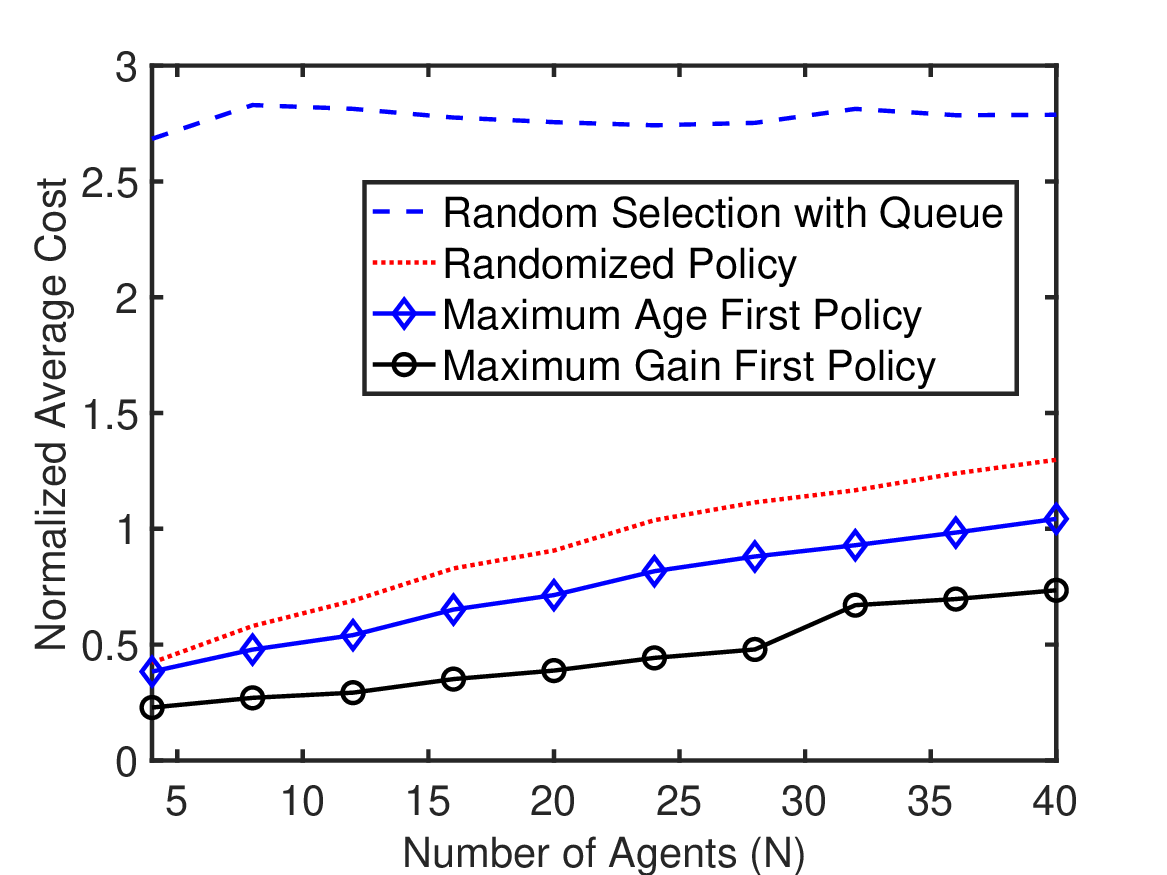}  
\vspace*{2.0mm} 
\caption{Normalized average penalty vs Number of Agents ($N$) where Number of channel is $M=1$ with success probability $0.95$.}\vspace{-0.0cm}
\label{figuresource}
\end{figure}

Figure \ref{figure3} demonstrates the penalty ($q_n (\delta, x)$) vs received observation ($x$) graph for different AoI ($\delta$) values. In this figure, when AoI is small, i.e., $\delta=1$, the penalty is high only at the two boundary regions which implies that we need to update frequently if $x$ is at any of the boundaries. Because near the boundaries, the uncertainty about the position of the agents is high, as a result, the scheduler has less situational awareness about the position of the agents. If $x$ is far from the boundary, then the situational awareness is good, and less frequent updating does not harm the system performance. With increasing $\delta$, the penalty graph spreads to the adjacent regions of the boundaries. Hence, the region that requires frequent updating is also increasing with increasing $\delta$. This intuition is crucial for designing an efficient status updating policy that can maximize situational awareness and eventually minimize the loss due to the unawareness of potential danger.

}

\section{Numerical Results} \label{simulation}


In this section, we evaluate the performance of the following policies:

\begin{itemize}

\item Randomized Policy: This policy randomly selects $M$ agents. To select $M$ agents out of $N$ agents, we use a built-in MATLAB function called ``randperm($N, M$)". 

\item Random Selection with Queue: The agents generate updates at every time slot and store in a FIFO queue. In our simulation, we consider the queue of each agent can store $1000$ updates. When the queue buffer is full, the oldest packet is dropped. At every time slot, $M$ agents are selected for transmission using the ``randperm($N, M$)" function. Whenever an agent is selected, the oldest packet from the queue of the agent is sent to the receiver.

The difference between Randomized policy and Random Selection with Queue is: In Randomized policy, whenever an agent is selected for transmission, a fresh update with AoI=$0$ is sent; In Random Selection with Queue, updates are generated at every time-slots and stored in a queue. When an agent is selected for transmission, the oldest packet from the queue is sent. 

\item Maximum Age First (MAF) Policy: This policy selects $M$ agents with the highest AoI.

\item Maximum Gain First (MGF) Policy: The policy provided in Algorithm \ref{algo1}.
 
\end{itemize}

We consider a safety-critical system with $N$ agents navigating within a region. This region is uniformly divided into a $20\time 20$ grid (20 rows and 20 columns). The observation $X_{n,t}$ for agent $n$ at time $t$ is represented by its position $x_n = (x_{n,1}, x_{n,2})$. 
The safety level $Y_{n,t}$ is categorized into three regions: safe, cautious, and dangerous. Rows 1 through 6 are designated as safe, rows 7 through 13 as cautious, and the remaining rows as dangerous. This setup differs from the one depicted in Figure \ref{figure3}, where all $400$ positions represent distinct states, leading to a state space of size $\tau \times 400$, where $\tau$ is the truncated AoI value. This large state space significantly increases the complexity of determining the value function and gain index. To mitigate this, we simplify the representation by considering only the row information for safety assessment. Consequently, the total number of states is reduced from $\tau \times 400$ to $\tau \times 20$. 

\begin{figure}
\vspace{-0.3cm}
\centering
\includegraphics[width=9cm]{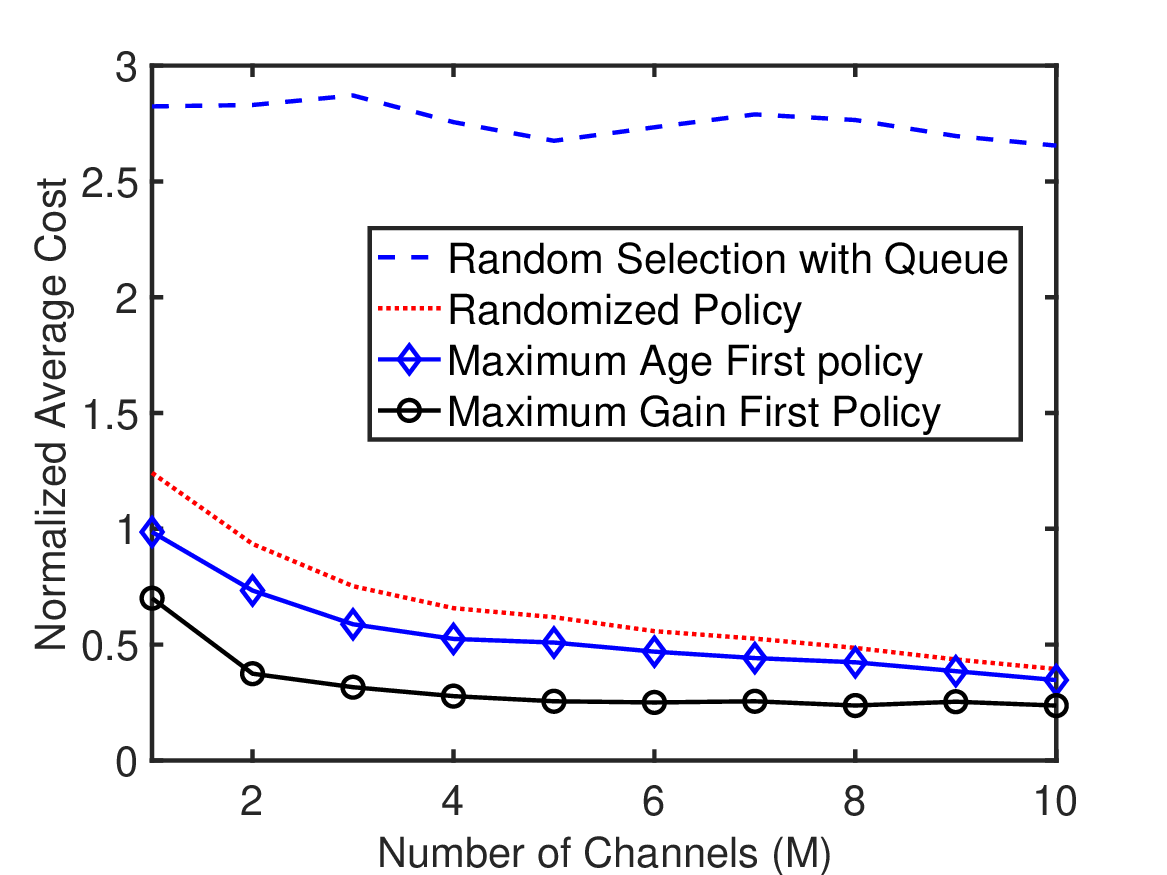}  
\vspace*{2.0mm} 
\caption{Normalized average penalty vs Number of channels ($M$) where Number of agents are $N=20$ with success probability $0.95$.}\vspace{-0.0cm}
\label{figurechannel}
\end{figure}

In this simulation, each agent $n$ can move randomly in any of the 
four directions: \emph{up, down, left}, and \emph{right}. We consider half of the agents move in \emph{up} with probability $0.3$, \emph{down} with probability $0.3$, \emph{left} with probability $0.2$, and \emph{right} with probability $0.2$. The rest of the agents move in \emph{up} with probability $0.05$, \emph{down} with probability $0.05$, \emph{left} with probability $0.45$, and \emph{right} with probability $0.45$. If an agent reaches a boundary, it stays in its current position. The loss incurred by agent $n$ is given by  $L_n(y, \hat y)$
which is defined as follows: $L_n $\emph{(dangerous, safe)} $=1000, L_n $\emph{(safe, dangerous)} $=5,
L_n $\emph{(cautious, safe)} $= 10, L_n $\emph{(safe, cautious)} $= 1, 1,L_n $\emph{(dangerous, cautious)} $=100$, $L_n $\emph{(cautious, dangerous)}$= 5$, and $L_n $\emph{(dangerous, dangerous)} $=L_n $\emph{(cautious, cautious)} $=L_n$\emph{(safe, safe)} $=0$.
In our simulation, the success probability is $0.95$.

\begin{figure}
\vspace{-0.3cm}
\centering
\includegraphics[width=9cm]{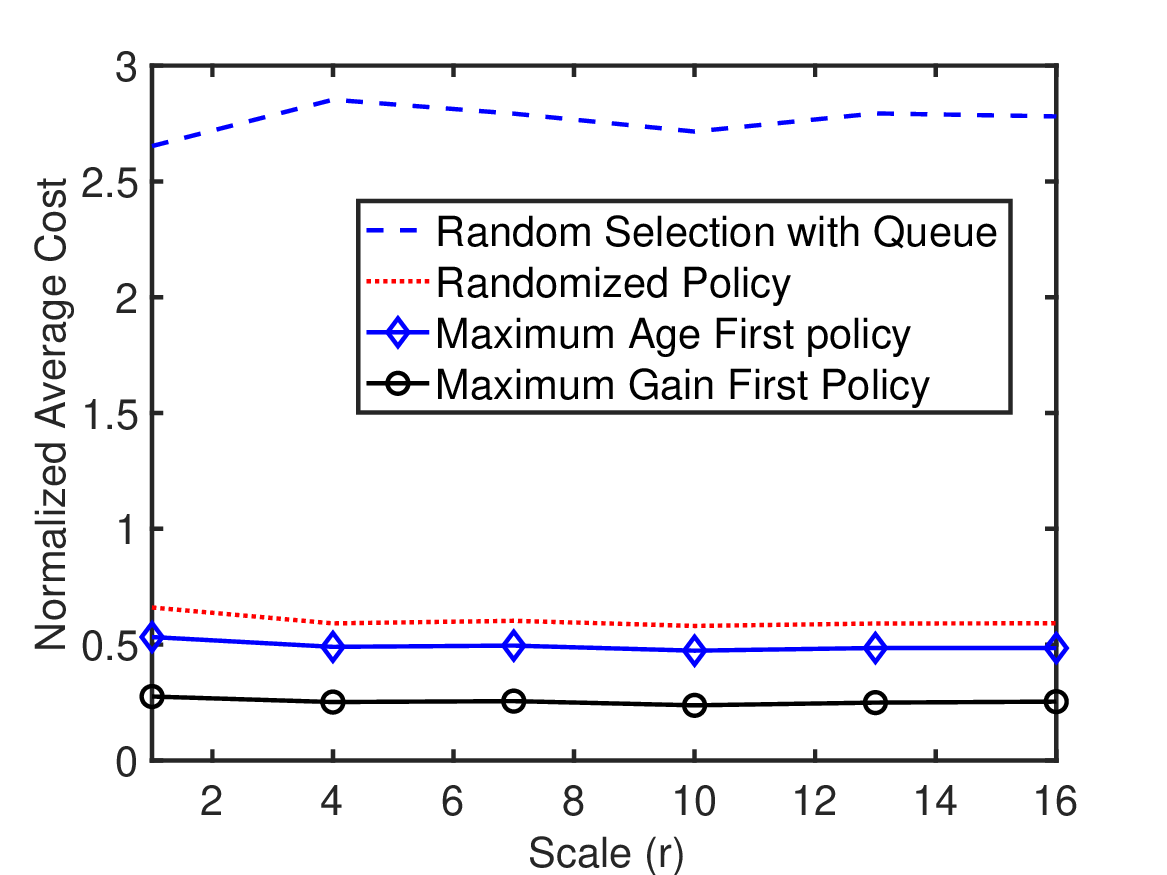} 
\vspace*{2.0mm} 
\caption{Normalized average penalty vs the scaling parameter ($r$) where the number of agents is $4r$ and the number of channels is $r$ with success probability $0.95$.}\vspace{-0.0cm}
\label{figureaysmp}
\vspace{-1em}
\end{figure}

Figure \ref{figuresource} illustrates the performance comparison of the four policies mentioned above as the number of agents increases. We consider two erasure channels in this simulation. The normalized average penalty (time-averaged penalty per agent) in Figure \ref{figuresource} is obtained by dividing time-average cost by the number of agents. From the figure, the MGF policy outperforms random selection with queue, randomized policy, and MAF policy. The performance of random selection with queue is always worse than the other three policies. This is because random selection with queue sends packets even when the previous packets has not yet been delivered, resulting in an extended waiting time in the queue and a significantly large AoI value at the receiver side. 
Furthermore, the randomized policy randomly selects two agents for sending updates and the MAF policy only utilizes the AoI in decision-making. In contrast, the MGF policy makes the decision in a smarter way by considering the gain index, which is a function of both the AoI and latest received observation. Hence, the MGF policy shows better performance than the other three policies. The performance gain of the MGF policy is up to $10.47$ times compared to random selection with queue, up to $2.33$ times compared to randomized policy, and up to $1.84$ times compared to MAF policy.

Figure \ref{figurechannel} illustrates the performance comparison of the four policies as the number of channels increases. We consider $20$ agents in this simulation. With the increase of the number of available channels for sending updates, the performance of the policies are getting better. However, because of the intelligent decision strategy by utilizing the AoI and the latest received observation, the MGF policy outperforms the other three policies. The performance gain of the MGF policy is up to $9.08$ times compared to random selection with queue, up to $2.5$ times compared to randomized policy, and up to $1.96$ times compared to MAF policy.


In Figure \ref{figureaysmp}, we consider the number of agents and the number of channels are $4r$ and $r$, respectively, where $r$ is a scaling parameter. The figure illustrates the normalized average cost as the scaling parameter $r$ increases. In this scenario, we observe that the MGF policy achieves the best performance compared to other policies under all simulated scaling parameters $r$.

Figure \ref{figuregain} demonstrates that the gain ($\alpha_n (\delta, x)$) vs the received observation ($x$) follows the same pattern what we observed from Figure \ref{figure3}. In this simulation, we consider $N=20$ and $M=10$. In Figure \ref{figure3}, we observe that frequent updating is required at the safety boundary regions and with the increase of AoI, this region requiring frequent updates expands. Similar pattern is observed from Figure \ref{figuregain} that illustrates that $\alpha_n (\delta, x)$ is high at the safety boundaries and we need frequent updating.


\begin{figure}
\centering
\includegraphics[width=9cm]{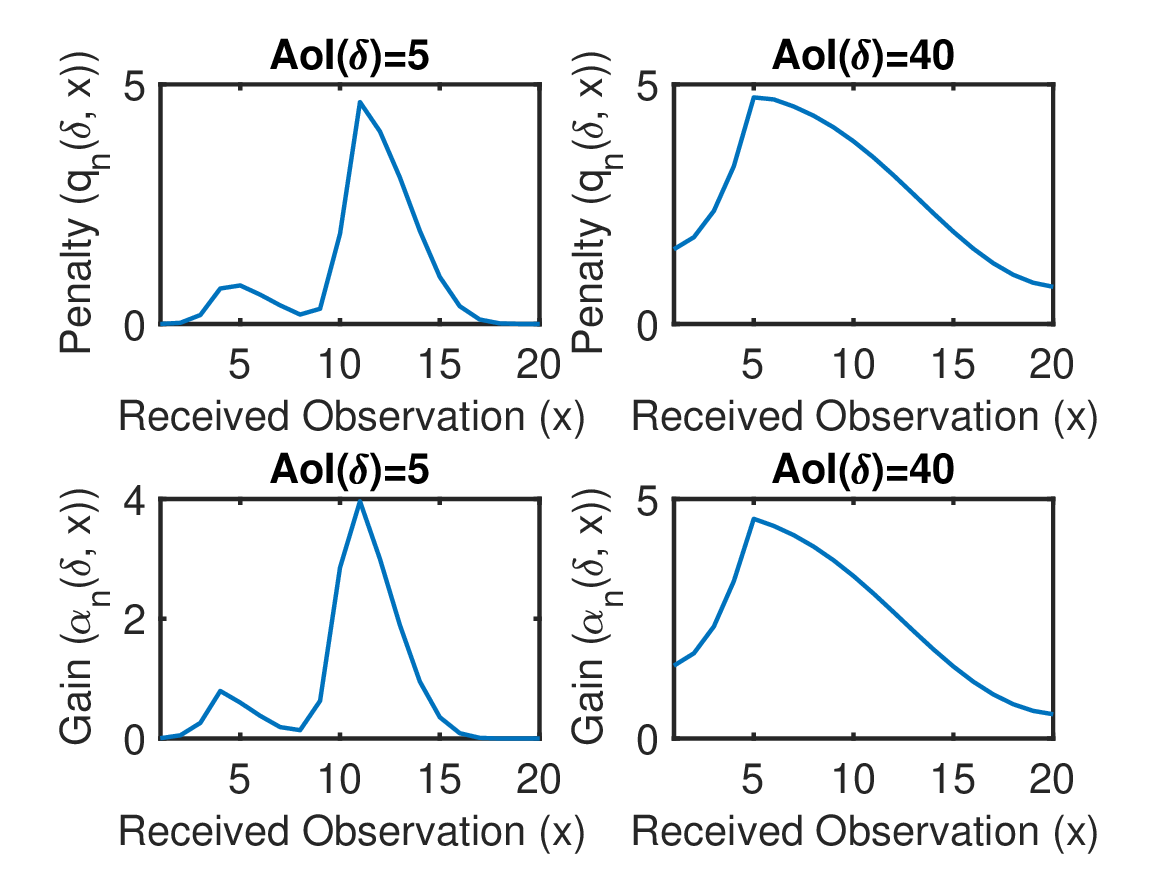}  
\vspace*{2.0mm} 
\caption{Penalty $q_n (\delta, x)$ vs received observation $x$ and Gain $\alpha_n (\delta, x)$ vs received observation $x$ for fixed AoI $\delta$. We consider number of agents $N=20$, number of channels $M=10$, success probability $0.95$, and $w_n =1$ for all agents.}\vspace{-0.0cm}
\label{figuregain}
\end{figure}

\section{conclusion}
In this study, we address the importance of situational awareness in safety-critical systems. We consider a pull-based status updating model and the formulated problem is an RMAB. The developed MGF policy requests updates more frequently when the received observation is near the safety boundaries. In addition, the MGF policy is not required to satisfy any indexability condition and is proven to be asymptotically optimal. In future we will study systems where agents dynamically enter and exit at any time. Another interesting direction is to consider a finite time horizon problem where there is a termination state while encountering a danger. 
\appendices
\section{Definitions of 0-1 loss, quadratic loss, logarithmic loss} \label{loss_definitions}

\textbf{0-1 loss:} The $0$-$1$ loss function assigns a loss of $0$ for incorrect estimation of a random variable $Y=y$, and $1$ otherwise. It is given by
and is given by an indicator function which equals to 1 if the true value $y$ and the estimated value $\hat y$ are equal to each other. 
$L_{0-1} (y, \hat y)$ is given by
\begin{align}
L_{0-1} (y, \hat y) = \mathbf{1}\{y \neq \hat y\},
\end{align}
where $\mathbf{1}\{y \neq \hat y\}$ is the indicator function for the event $\{y \neq \hat y\}$.

\textbf{Quadratic loss:} The quadratic loss function quantifies the error between the true value of a random variable $Y=y$ and and its estimated value $\hat y$ by calculating the square of their difference. It is given by
\begin{align}
L_2 (y, \hat y) = (y - \hat y)^2. 
\end{align}

\textbf{Logarithmic loss:} The log-loss function $L_{\log} (y, P_Y )$ is the negative log-likelihood of the true value $Y = y$ which is given by
\begin{align}
L_{\log} (y, P_Y ) = - \log P_Y (y)
\end{align}
where the action $a = P_Y$ is a distribution of $Y$.

\section{Proof of Theorem \ref{thm_sufficient_statistic}} \label{ptheorem1}

To prove (i), we need to show that $Y_{n, t}$ is conditionally independent of the history $\mathbf{H}_{n, t}$ given the latest received observation $X_{n, t-\Delta_n(t)}$ and its AoI $\Delta_n (t)$ at time $t$: 
\begin{align} \label{independence_proof}
\mathbb{P} (Y_{n, t} | \mathbf{H}_{n, t}) 
=& \mathbb{P} (Y_{n, t} | \big(X_{n,\tau-\Delta_n(\tau)}, \Delta_n (\tau), \mu_n (\tau-1), \gamma_n (\tau-1)\big)_{\tau=0}^{t}) \nonumber\\
=& \mathbb{P} (Y_{n, t}|(\Delta_n (t), X_{n, t-\Delta_n(t)})).
\end{align}
We will prove \eqref{independence_proof} in two steps: (1) First, we show that $Y_{n, t}$ is conditionally independent of the scheduling decisions $(\mu_n (\tau))_{\tau =0} ^ {t-1}$ and the delivery indicators $(\gamma_n (\tau))_{\tau =0} ^ {t-1}$ given the AoI history $(\Delta_n (\tau))_{\tau =0} ^ {t}$; (2) Next, we show that $Y_{n, t}$ is conditionally independent of the $(\Delta_n (\tau), X_{n, \tau - \Delta_n (\tau)})_{\tau =0} ^ {t-1}$  given the AoI $(\Delta_n (t), X_{n, t-\Delta_n (t)})$. 

The safety level $Y_{n, t}$ depends on $(X_{n,\tau-\Delta_n(\tau)}, \Delta_n (\tau))_{\tau=0}^{t})$ and to determine $(\Delta_n (\tau))_{\tau=0}^{t})$, we need to know the scheduling decisions $(\mu_n (\tau))_{\tau =0} ^ {t-1}$ and the delivery indicators $(\gamma_n (\tau))_{\tau =0} ^ {t-1}$, as defined in \eqref{age}. However, if the AoI history $(\Delta_n (\tau))_{\tau =0} ^ {t}$ is given, then the safety level $Y_{n, t}$ does not depend on the the scheduling decisions $(\mu_n (\tau))_{\tau =0} ^ {t-1}$ and the delivery indicators $(\gamma_n (\tau))_{\tau =0} ^ {t-1}$. Hence, $Y_{n, t}$ is conditionally independent of the scheduling decisions $(\mu_n (\tau))_{\tau =0} ^ {t-1}$ and the delivery indicators $(\gamma_n (\tau))_{\tau =0} ^ {t-1}$ given the AoI history $(\Delta_n (\tau))_{\tau =0} ^ {t}$.
Therefore, we can write
\begin{align} \label{independence_proof_1}
\mathbb{P} (Y_{n, t} | \mathbf{H}_{n, t})
=\mathbb{P} (Y_{n, t} | \big(X_{n,\tau-\Delta_n(\tau)}, \Delta_n (\tau)\big)_{\tau=0}^{t}).
\end{align}

Due to the Markov property of the underlying process $Y_{n, t} \leftrightarrow X_{n,t} \leftrightarrow X_{n, t-1}\leftrightarrow \ldots$, we have the following Markov chain given $\Delta_n (t)$: $Y_{n, t} \leftrightarrow X_{n,t-\Delta_n (t)} \leftrightarrow X_{n, (t-1)-\Delta_n (t-1)}\leftrightarrow \ldots$.
Hence, $Y_{n, t}$ is independent of $(\Delta_{n}(\tau), X_{n, \tau-\Delta_n(\tau)})_{\tau=0}^{t-1}$ given $\Delta_n (t)$ and $X_{n, t-\Delta_n (t)}$. Considering these facts, \eqref{independence_proof} follows from \eqref{independence_proof_1}. Therefore, we can write
\begin{align}\label{sufficientestimator}
\inf_{y \in  \mathcal Y_n} \mathbb{E} [L(Y_{n,t}, y)|\mathbf{H}_{n, t}]
=\inf_{y \in  \mathcal Y_n} \!\!\!\mathbb{E} [L(Y_{n,t}, y)|\Delta_n (t), X_{n, t-\Delta_n (t)}].
\end{align}
and because $\mathcal F_n$ is the set of all possible functions that maps any input from $(\Delta_n (t), X_{n, t-\Delta_n (t)})$ to $\mathcal Y_n$. 
This proves (i).

To prove (ii), first we show that

\begin{align} \label{exp}
    &\mathbb E\bigg[\sum_{n=1}^{N} L_n(Y_{n, t}, \phi_n (\mathbf H_{n,t}))-\mathsf{L}_{\text{opt}}\bigg|(\mathbf H_{n,t})_{n=1}^{N}\bigg]\nonumber\\
    &=\sum_{n=1}^{N} \mathbb E\bigg[ L_n(Y_{n, t}, \phi_n (\mathbf H_{n,t}))-\mathsf{L}_{\text{opt}}\bigg|\mathbf H_{n,t}\bigg],
\end{align}
which holds because given the scheduling decisions $(\mu_n (\tau))_{\tau=0}^{t-1}$, the delivery indicators $(\gamma_n(t))_{\tau=0}^{t-1}$ are independent across both agents and time slots and the statuses \{$X_{n, t}, t = 0,1,2,\ldots$\} and \{$X_{m, t}, t = 0,1,2,\ldots$\} are independent for all $n\neq m$. 

Now, by utilizing Theorem \ref{thm_sufficient_statistic} (i) in \eqref{exp}, we get
\begin{align} \label{exp_1}
    & \mathbb E\!\bigg[\! L_n(Y_{n, t}, \phi_n (\mathbf H_{n,t}))-\mathsf{L}_{\text{opt}}\bigg|\mathbf H_{n,t}\bigg] \nonumber\\
    =& \mathbb E\!\bigg[\! L_n(Y_{n, t}, f_n (\Delta_n (t), \!X_{n, t-\Delta_n (t)})\!-\!\mathsf{L}_{\text{opt}}\bigg|\Delta_n (t), \!X_{n, t-\Delta_n (t)}\!\bigg]\!.
\end{align}

Hence, the value function of the average cost MDP \eqref{problem}-\eqref{constraint} under any given policy $\pi=(\mu_n(0), \mu_n(1), \ldots)_{n=1}^{N}$ can be reduced to
\begin{align} \label{Jpi}
& \sum_{t=0}^{\infty} \sum_{n=1}^{N} \mathbb E\bigg[ L_n(Y_{n, t}, \phi_n (\mathbf H_{n,t}))-\mathsf{L}_{\text{opt}}\bigg|(\mathbf H_{n,t})_{n=1}^{N}\bigg] \nonumber\\
=& \sum_{t=0}^{\infty} \sum_{n=1}^{N} \!\mathbb E\!\bigg[L_n (Y_{n, t}, f_n (\Delta_n (t), X_{n, t-\Delta_n (t)})-\mathsf{L}_{\text{opt}}\bigg|\Delta_n (t), X_{n, t-\Delta_n (t)}\bigg],
\end{align}
where \eqref{Jpi} holds from \eqref{exp_1}.
Therefore, from \cite[Chapter 4.3]{bertsekas2011dynamic}, Theorem \ref{thm_sufficient_statistic} (ii) follows.

\section{Proof of Lemma \ref{lemmaoptimalestimator}} \label{proof_lemma_joint}

Given any policy $\pi \in \Pi$, problem \eqref{problem}-\eqref{constraint} can be written as
\begin{align}
\!\mathsf{L}_{\text{opt}}\! &=\! \inf_{\phi \in \Phi}\! \limsup_{T \to \infty}\! \sum_{n=1}^{N} \! \frac{1}{T} \!\sum_{t=0}^{T-1} \!\mathbb{E} \big[L_n(Y_{n, t}, \phi_n (\mathbf{H}_{n, t}))\big] \label{problem_proof}\\
& \quad \limsup_{T \to \infty}\! \sum_{n=1}^{N} \! \frac{1}{T} \!\sum_{t=0}^{T-1} \inf_{\phi \in \Phi}\! \mathbb{E} \big[L_n(Y_{n, t}, \phi_n (\mathbf{H}_{n, t}))\big]. \label{constraint_proof}
\end{align}

According to Theorem \ref{thm_sufficient_statistic} (i), $(\Delta_n (t), X_{n, t-\Delta_n (t)})$ is a sufficient statistic of $\mathbf{H}_{n, t}$ for estimating $Y_{n, t}$, the optimal estimation problem can be written as
\begin{align}
\inf_{f_n \in \mathcal Y_n} \mathbb{E} \big[\!L_n(Y_{n, t}, f_n (\Delta_n (t), X_{n, t-\Delta_n (t)})\!)\!\big],    
\end{align}
where $f_n$ is the set of all estimator functions which maps any $(\delta, x) \in \mathbb Z^{+} \times \mathcal X_n$ to $\mathcal Y_n$. Therefore, we can write
\begin{align}
 & \sum_{(\delta, x)} P_{\Delta_n(t),X_{n, t-\Delta_n (t)}}(\delta,x) \inf_{f_n (x, \delta) \in \mathcal{Y}_n} \!\!\!\!\!\mathbb{E}_{P_{Y_{n,t}}|\Delta_n (t) = \delta, X_{n, t-\delta} = x} [L_n (Y, f_n (\delta, x)] \\
 =& \sum_{(\delta, x)}  P_{\Delta_n(t),X_{n, t-\Delta_n (t)}}(\delta,x) \inf_{y_n \in \mathcal{Y}_n} \mathbb{E}_{P_{Y_{n,t}}|\Delta_n (t) = \delta, X_{n, t-\delta} = x} [L_n (Y, y_n)].
\end{align}
This completes the proof of Lemma \ref{lemmaoptimalestimator}.

\section{Proof of Lemma \ref{lemma_entropy}} \label{proof_lemma_entropy}

By definition  \eqref{L_cond_en}, we have 
\begin{align}
H_L (Y_{n, t} |\Delta_n(t), X_{n, t-\Delta_n(t)})
=&\sum_{(\delta, x) \in \mathbb Z \times \mathcal X_n} P_{\Delta_n(t),X_{n, t-\Delta_n (t)}}(\delta,x) \times \inf_{f_n(\delta, x) \in \mathcal Y_n} \mathbb{E}_{Y \sim P_{Y_{n, t}|\Delta_n (t) = \delta, X_{n, t-\delta}=x}} \big[L_n(Y, f_n (\delta, x))\big]\nonumber\\
=&\sum_{(\delta, x) \in \mathbb Z \times \mathcal X_n} P_{\Delta_n(t),X_{n, t-\delta}}(\delta,x) \inf_{f_n(\delta, x) \in \mathcal Y_n} \mathbb{E}_{Y \sim P_{Y_{n, t}|\Delta_n (t) = \delta, X_{n, t-\delta}=x}} \big[L_n(Y, f_n (\delta, x))\big].
\end{align}

\section{Proof of Theorem \ref{optimality}} \label{proof_optimality}

{ We first present a set of LP-based priority policies that achieve asymptotically optimality under the uniform global attractor condition in Definition \ref{def2}. Subsequently, we demonstrate that $\pi_{\text{gain}}$ belongs to this set of priority policies. Let $U_{\delta, x}^{n, \mu}(t)$ be the number of class $n$ bandits in state ${\delta, x}$  taking action $\mu$. We define $u^{n, \mu}_{\delta, x}$ as the expected number of class $n$ bandits in state ${\delta, x}$  taking action $\mu$, given by
\begin{align}
 u^{n, \mu}_{\delta, x}=\limsup_{T\rightarrow \infty}\sum_{t=0}^{T-1} \frac{1}{T} \mathbb E [ U^{n, \mu}_{\delta, x}(t)].
\end{align}
Let $\bar u^{n, \mu}_{\delta, x}$ be the optimal state action frequency of the Lagrangian problem \eqref{decomposed} in which $\lambda=\lambda^*$ is the optimal dual Lagrangian multiplier. 

Define the following sets
\begin{align}
\mathcal{S}_{+} ^n = \{(\delta, x) : \bar u_{\delta, x} ^{n, 1} > 0 {\thinspace}~ {\text{and}} ~{\thinspace} \bar u_{\delta, x} ^{n, 0} = 0\}, \\
\mathcal{S}_{0} ^n = \{(\delta, x) : \bar u_{\delta, x} ^{n, 1} > 0 {\thinspace}~ {\text{and}} ~{\thinspace} \bar u_{\delta, x} ^{n, 0} > 0\}, \\
\mathcal{S}_{-} ^n = \{(\delta, x) : \bar u_{\delta, x} ^{n, 1} = 0 {\thinspace}~ {\text{and}} ~{\thinspace} \bar u_{\delta, x} ^{n, 0} \geq 0\}.
\end{align}

\begin{definition} \label{def_4}
\textbf{LP-based Priority Policies}. \cite{verloop2016asymptotically}, \cite{gast2023linear} We define a set $\Pi_{\text{LP-Priority}}$ that consists of priority policies that satisfy the following conditions:

i) A class-$k$ bandit in state $(\delta_k, x_k) \in \mathcal{S}_{+} ^k$ is given higher priority than a class-$j$ bandit in state $(\delta_j, x_j) \in \mathcal{S}_{0} ^j$.

ii)  A class-$k$ bandit in state $(\delta_k, x_k) \in \mathcal{S}_{0} ^k$ is given higher priority than a class-$j$ bandit in state $(\delta_j, x_j) \in \mathcal{S}_{-} ^j$.

iii) If a class-$k$ bandit is activated, then $\mu=1$ is chosen; otherwise, $\mu=0$ is chosen.
\end{definition}

\begin{lemma} \label{lemma_3}
For any policy $\pi \in \Pi_{\text{LP-Priority}}$, if the uniform global attractor condition in Definition \ref{def2} is satisfied, then the policy is asymptotically optimal. 
\end{lemma}

By utilizing \cite[Theorem 13]{gast2023linear}, we can obtain Lemma \ref{lemma_3}. The results in \cite{gast2023linear} hold for RMAB problem with equality constraint. Our problem \eqref{problem}-\eqref{constraint} considers an inequality constraint \eqref{constraint}. In order to use the results from  \cite{gast2023linear}, we introduce dummy bandits and obtain an equivalent problem \eqref{problem_dummy}-\eqref{constraint_dummy_2} with equality constraint \eqref{constraint_dummy}. 

Following \cite[Proposition 14]{gast2023linear}, for class-$k$ bandits

1. For any class-$k$ bandit, $(\delta_k, x_k) \in \mathcal{S}_{+}$ implies
\begin{align} \label{priority_1}
\alpha_n (\delta_k, x_k) > 0.
\end{align}

2. For any class-$k$ bandit, $(\delta_k, x_k) \in \mathcal{S}_{0}$ implies
\begin{align}
\alpha_n (\delta_k, x_k) = 0.
\end{align}

3. For any class-$k$ bandit, $(\delta_k, x_k) \in \mathcal{S}_{-}$ implies
\begin{align} \label{priority_3}
\alpha_n (\delta_k, x_k) < 0.
\end{align}

Because the policy $\pi_{\text{gain}}$ activates exactly $rM$ bandits with the highest gain and if a bandit $k$ is activated, $\pi_{\text{gain}}$ chooses $\mu=1$, we can deduce form \eqref{priority_1}-\eqref{priority_3} and Definition \ref{def_4} that $\pi_{\text{gain}}$ belongs to $\Pi_{\text{LP-Priority}}$. This concludes the proof.

}

\section{Special Case: Single-Source, Single-Channel} \label{single_source}

Let us consider a special case with $N = M = 1$, where the system has one source and one channel. Then, problem \eqref{problem_eq_1}-\eqref{constraint_eq_1} reduces to 
\begin{align}
\!\mathsf{L}_{1, \text{opt}}\! =\!\! & \inf_{\pi_1 \in \Pi_1}\! \limsup_{T \to \infty} \mathbb{E} \bigg[\frac{1}{T} \!\sum_{t=0}^{T-1} H_L (Y_{1, t}|\Delta_1 (t), X_{1, t-\Delta_1 (t)}) \!\bigg]. \label{problem_entropy}
\end{align}

In this sequel, we have the following theorem.

\begin{theorem} \label{theorem1}
For single-source, single-channel case, the optimal policy $\pi_1$ in \eqref{problem_entropy} chooses the active action at every time slot $t$.
\end{theorem}


Theorem \ref{theorem1} states that for single-source, single-channel case, it is always better to send. The proof of Theorem \ref{theorem1} is presented below:

Problem \eqref{problem_entropy} 
is an MDP that can be solved by Dynamic programming \cite{bertsekas2011dynamic}. The optimal policy associated with agent $1$ satisfies the following Bellman optimality equation:
\begin{align} \label{bellman}
J_1 (\delta, x) = H_L (Y_{1,\delta}|X_{1, 0}=x) - \mathsf{L}_{1, \text{opt}}  + \min\{J_1 (\delta+1, x), (1-p_1) J_1 (\delta+1, x) + p_1 \mathbb{E} [J_1(1, X_{1, 0})|X_{1, -\delta} =x]\},
\end{align}
where $J_1 (\delta, x)$ is the value function associated with state $(\delta, x)$ and $\mathsf{L}_{1, \text{opt}}$ is the optimal value of \eqref{problem_entropy}. 

As explained in \cite{bertsekas2011dynamic}, the optimal value function can be derived by using value iteration and the sequence of value functions $J_{1, k} (\delta, x)$ can be written as
\begin{align} \label{value_step}
J_{1, k+1} (\delta, x)
= H_L (Y_{1,\delta}|X_{n, 0}=x) - \mathsf{L}_{n, \text{opt}} + \min\{J_{1, k} (\delta+1, x), (1-p_1) J_{1,k} (\delta+1, x) + p_1 \mathbb{E} [J_{1, k} (1, X_{1, 0})|X_{1, -\delta} =x]\},
\end{align}
which converges to $\lim_{k \to \infty} J_{1,k} = J_1$ for any $J_{1, 0}$. After some rearrangements, we can write \eqref{value_step} as
\begin{align} \label{value_step_1}
J_{1, k+1} (\delta, x)
= H_L (Y_{1,\delta}|X_{1, 0}=x) - \mathsf{L}_{1, \text{opt}} +J_{1,k} (\delta+1, x) + p_1 \min\{0, - J_{1, k} (\delta+1, x) + \mathbb{E} [J_{1, k} (1, X_{1, 0})|X_{1, -\delta} =x]\}.
\end{align}

In this sequel, we introduce the following useful lemma which illustrates that more information reduces the $L$-conditional entropy.

\begin{lemma} \label{lemma1}
For random variables $X, Y,$ and $Z$, it holds that $H_L (Y|Z=z) \geq H_L (Y|X, Z=z)$, where
\begin{align}
H_L (Y|Z=z) &= \min_{a \in \mathcal{A}} \mathbb{E} [L(Y,a)|Z=z], \label{entropy_eq_1}\\
H_L (Y|X, Z\!=\!z) \!\!&= \!\!\sum_{x \in \mathcal{X}} P(X\!=\!x|Z\!=\!z) H_L (Y|X\!=\!x, Z\!=\!z). \label{entropy_eq_2}
\end{align}
\end{lemma}

\begin{proof}
See Appendix \ref{proof_lemma1}.
\end{proof}

Using Lemma \ref{lemma1}, we get that for any $k$, $J_{1, k} (\delta+1, x) \geq \mathbb{E} [J_{1,k} (1, X_{1,0})|X_{1, -\delta} = x]$ which implies that the penalty for not sending at iteration step $k$ is higher than sending. Therefore, taking the active action is beneficial to reduce the penalty. One interesting observation from \eqref{value_step_1} is that each time a packet is successfully delivered with probability $p_1$, a new piece of information about the sensor signal value is added with the existing information $(X_{1, t-\delta} =x)$ (see the term $\mathbb{E} [J_{1,k} (1, X_{1,0})|X_{1, t-\delta} =x]$ in \eqref{value_step}). This new information plays a crucial role in reducing the system penalty and hence benefits the system through sending. This completes the proof of Theorem \ref{theorem1}.

\section{Proof of Lemma \ref{lemma1}} \label{proof_lemma1}

From the definition of $L$-conditional entropy in \eqref{L_cond_en}, we get that
\begin{align}
& H_L (Y|Z=z) \nonumber\\
&= \min_{a \in \mathcal{A}} \mathbb{E} [L(Y,a) | Z=z] \label{def} \\ 
&= \min_{a \in \mathcal{A}} \sum_{y \in \mathcal{Y}} P(Y=y | Z=z) L(y,a) \nonumber\\
&= \min_{a \in \mathcal{A}} \!\sum_{y \in \mathcal{Y}} \!\sum_{x \in \mathcal{X}} \!P(Y\!=\!y |X\!=\!x, Z\!=\!z) P(X\!=\!x|Z\!=\!z) L (y,a) \nonumber\\
&= \min_{a \in \mathcal{A}}\! \sum_{x \in \mathcal{X}}\! P(X\!=\!x | Z\!=\!z) \sum_{y \in \mathcal{Y}} P(Y\!=\!y | X\!=\!x, Z\!=\!z) L(y,a) \nonumber\\
&\geq \!\sum_{x \in \mathcal{X}} \!P(X\!=\!x | Z\!=\!z) \min_{a \in \mathcal{A}} \!\sum_{y \in \mathcal{Y}} \!P(Y\!=\!y | X\!=\!x, Z\!=\!z) L(y,a) \label{proof_eq_2}
\end{align}
where \eqref{proof_eq_2} holds because $\min (f(w)+g(w)) \geq \min f(w) + \min g(w)$ for all $w$. Continuing from \eqref{proof_eq_2}, we get that
\begin{align}
H_L (Y|Z=z)
\geq  \sum_{x \in \mathcal{X}} P(X=x | Z=z) \min_{a \in \mathcal{A}} \mathbb{E} [L(Y,a) | X=x, Z=z]. \label{proof_eq_3}
\end{align}
Utilizing \eqref{entropy_eq_1}, we obtain that \cite{dawid1998coherent, farnia2016minimax, Shisher2022}
\begin{align}
\!\!\!H_L (Y|X=x, Z=z)
= \min_{a \in \mathcal{A}} \mathbb{E} [L(Y,a) | X=x, Z=z]. \label{entropy_eq_3}
\end{align}
Substituting \eqref{entropy_eq_3} into \eqref{proof_eq_3} yields
\begin{align}
H_L (Y|Z=z)\!\!\geq & \sum_{x \in \mathcal{X}} P(X=x|Z=z) H_L (Y|X=x, Z=z), \nonumber\\
=& H_L (Y|X, Z=z), \label{entropy_eq_4}
\end{align}
where \eqref{entropy_eq_4} follows from \eqref{entropy_eq_2}. This completes the proof.

\bibliographystyle{IEEEtran}
\bibliography{ref,ref1,ref_2,sueh}

\begin{thebibliography}{10}
\providecommand{\url}[1]{#1}
\csname url@samestyle\endcsname
\providecommand{\newblock}{\relax}
\providecommand{\bibinfo}[2]{#2}
\providecommand{\BIBentrySTDinterwordspacing}{\spaceskip=0pt\relax}
\providecommand{\BIBentryALTinterwordstretchfactor}{4}
\providecommand{\BIBentryALTinterwordspacing}{\spaceskip=\fontdimen2\font plus
\BIBentryALTinterwordstretchfactor\fontdimen3\font minus
  \fontdimen4\font\relax}
\providecommand{\BIBforeignlanguage}[2]{{%
\expandafter\ifx\csname l@#1\endcsname\relax
\typeout{** WARNING: IEEEtran.bst: No hyphenation pattern has been}%
\typeout{** loaded for the language `#1'. Using the pattern for}%
\typeout{** the default language instead.}%
\else
\language=\csname l@#1\endcsname
\fi
#2}}
\providecommand{\BIBdecl}{\relax}
\BIBdecl

\bibitem{Ornee_MILCOM}
T.~Z. Ornee, M.~K.~C. Shisher, C.~Kam, and Y.~Sun, ``Context-aware status
  updating: Wireless scheduling for maximizing situational awareness in
  safety-critical systems,'' in \emph{IEEE MILCOM}, 2023, pp. 194--200.

\bibitem{grau2017industrial}
A.~Grau, M.~Indri, L.~L. Bello, and T.~Sauter, ``Industrial robotics in factory
  automation: From the early stage to the internet of things,'' in \emph{IEEE
  IECON}, 2017, pp. 6159--6164.

\bibitem{abdulmalek2022iot}
S.~Abdulmalek, A.~Nasir, W.~A. Jabbar, M.~A. Almuhaya, A.~K. Bairagi, M.~A.-M.
  Khan, and S.-H. Kee, ``{I}o{T}-based healthcare-monitoring system towards
  improving quality of life: A review,'' in \emph{Healthcare}, vol.~10, no.~10,
  2022, p. 1993.

\bibitem{adeel2019survey}
A.~Adeel, M.~Gogate, S.~Farooq, C.~Ieracitano, K.~Dashtipour, H.~Larijani, and
  A.~Hussain, ``A survey on the role of wireless sensor networks and iot in
  disaster management,'' \emph{Geological disaster monitoring based on sensor
  networks}, pp. 57--66, 2019.

\bibitem{KaulYatesGruteser-Infocom2012}
S.~Kaul, R.~D. Yates, and M.~Gruteser, ``Real-time status: How often should one
  update?'' in \emph{IEEE INFOCOM}, 2012.

\bibitem{Liu_TIT}
K.~Liu and Q.~Zhao, ``Indexability of restless bandit problems and optimality
  of whittle index for dynamic multichannel access,'' \emph{IEEE Trans. Inf.
  Theory}, vol.~56, no.~11, pp. 5547--5567, 2010.

\bibitem{Ouyang_TMC}
W.~Ouyang, A.~Eryilmaz, and N.~B. Shroff, ``Low-complexity optimal scheduling
  over time-correlated fading channels with arq feedback,'' \emph{IEEE Trans.
  Mob. Comput.}, vol.~15, no.~9, pp. 2275--2289, 2016.

\bibitem{Peter_ITC}
P.~Jacko and S.~S. Villar, ``Opportunistic schedulers for optimal scheduling of
  flows in wireless systems with arq feedback,'' in \emph{IEEE ITC}, 2012, pp.
  1--8.

\bibitem{Ouyang_Infocom}
W.~Ouyang, A.~Eryilmaz, and N.~B. Shroff, ``Asymptotically optimal downlink
  scheduling over markovian fading channels,'' in \emph{IEEE INFOCOM}, 2012,
  pp. 1224--1232.

\bibitem{Javidi_TIT}
S.~H.~A. Ahmad, M.~Liu, T.~Javidi, Q.~Zhao, and B.~Krishnamachari, ``Optimality
  of myopic sensing in multichannel opportunistic access,'' \emph{IEEE Trans.
  Inf. Theory}, vol.~55, no.~9, pp. 4040--4050, 2009.

\bibitem{Ansell_2003}
P.~Ansell, K.~Glazebrook, J.~Niño-Mora, and M.~O'Keeffe, ``Whittle's index
  policy for a multi-class queueing system with convex holding costs,''
  \emph{Math. Meth. Oper. Res.}, vol.~57, pp. 21--39, 04 2003.

\bibitem{Neely_2010}
C.-p. Li and M.~J. Neely, ``Exploiting channel memory for multi-user wireless
  scheduling without channel measurement: Capacity regions and algorithms,'' in
  \emph{ACM MobiHoc}, 2010, pp. 50--59.

\bibitem{Murugesan_2012}
S.~Murugesan, P.~Schniter, and N.~B. Shroff, ``Multiuser scheduling in a
  markov-modeled downlink using randomly delayed arq feedback,'' \emph{IEEE
  Trans. Inf. Theory}, vol.~58, no.~2, pp. 1025--1042, 2012.

\bibitem{chen2021scheduling}
Y.~Chen and A.~Ephremides, ``Scheduling to minimize age of incorrect
  information with imperfect channel state information,'' \emph{Entropy},
  vol.~23, no.~12, p. 1572, 2021.

\bibitem{chen2022index}
G.~Chen and S.~C. Liew, ``An index policy for minimizing the
  uncertainty-of-information of {Markov} sources,'' \emph{arXiv preprint
  arXiv:2212.02752}, 2022.

\bibitem{Chen_2022}
G.~Chen, S.~C. Liew, and Y.~Shao, ``Uncertainty-of-information scheduling: A
  restless multiarmed bandit framework,'' \emph{IEEE Trans. Inf. Theory},
  vol.~68, no.~9, pp. 6151--6173, 2022.

\bibitem{bertsekas2011dynamic}
D.~P. Bertsekas \emph{et~al.}, ``Dynamic programming and optimal control, 4th
  edition,'' \emph{Belmont, MA: Athena Scientific}, vol.~1, 2011.

\bibitem{tripathi2019whittle}
V.~Tripathi and E.~Modiano, ``A {Whittle} index approach to minimizing
  functions of age of information,'' in \emph{IEEE Allerton}, 2019, pp.
  1160--1167.

\bibitem{ornee2023whittle}
T.~Z. Ornee and Y.~Sun, ``A {Whittle} index policy for the remote estimation of
  multiple continuous {Gauss-Markov} processes over parallel channels,'' in
  \emph{ACM MobiHoc}, 2023, p. 91–100.

\bibitem{Shisher_2024}
M.~K.~C. Shisher, Y.~Sun, and I.-H. Hou, ``Timely communications for remote
  inference,'' \emph{IEEE/ACM Transactions on Networking}, vol.~32, no.~5, pp.
  3824--3839, 2024.

\bibitem{sun2018information}
Y.~Sun and B.~Cyr, ``Information aging through queues: A mutual information
  perspective,'' in \emph{2018 IEEE 19th International Workshop on Signal
  Processing Advances in Wireless Communications (SPAWC)}.\hskip 1em plus 0.5em
  minus 0.4em\relax IEEE, 2018, pp. 1--5.

\bibitem{sun2019sampling}
------, ``Sampling for data freshness optimization: Non-linear age functions,''
  \emph{J. Commun. Netw.}, vol.~21, no.~3, pp. 204--219, 2019.

\bibitem{VoI_Kosta}
A.~Kosta, N.~Pappas, A.~Ephremides, and V.~Angelakis, ``Age and value of
  information: Non-linear age case,'' in \emph{IEEE ISIT}, 2017, pp. 326--330.

\bibitem{wang2022framework}
Z.~Wang, M.-A. Badiu, and J.~P. Coon, ``A framework for characterizing the
  value of information in hidden {Markov} models,'' \emph{IEEE Trans. Inf.
  Theory}, vol.~68, no.~8, pp. 5203--5216, 2022.

\bibitem{soleymani2016optimal}
T.~Soleymani, S.~Hirche, and J.~S. Baras, ``Optimal self-driven sampling for
  estimation based on value of information,'' in \emph{IEEE WODES}, 2016, pp.
  183--188.

\bibitem{Shisher2022}
M.~K.~C. Shisher and Y.~Sun, ``How does data freshness affect real-time
  supervised learning?'' in \emph{ACM MobiHoc}, 2022, pp. 31--40.

\bibitem{shisher2021age}
M.~K.~C. Shisher, H.~Qin, L.~Yang, F.~Yan, and Y.~Sun, ``The age of correlated
  features in supervised learning based forecasting,'' in \emph{IEEE INFOCOM
  Workshops}, 2021, pp. 1--8.

\bibitem{shisher2023learning}
M.~K.~C. Shisher, B.~Ji, I.-H. Hou, and Y.~Sun, ``Learning and communications
  co-design for remote inference systems: Feature length selection and
  transmission scheduling,'' \emph{IEEE J. Sel. Areas Inf. Theory}, vol.~4, pp.
  524--538, 2023.

\bibitem{Sun_TIT_2017}
Y.~Sun, E.~Uysal-Biyikoglu, R.~D. Yates, C.~E. Koksal, and N.~B. Shroff,
  ``Update or wait: How to keep your data fresh,'' \emph{IEEE Trans. Inf.
  Theory}, vol.~63, no.~11, pp. 7492--7508, 2017.

\bibitem{Bedewy_2021}
A.~M. Bedewy, Y.~Sun, S.~Kompella, and N.~B. Shroff, ``Optimal sampling and
  scheduling for timely status updates in multi-source networks,'' \emph{IEEE
  Trans. Inf. Theory}, vol.~67, no.~6, pp. 4019--4034, 2021.

\bibitem{Ornee2021}
T.~Z. Ornee and Y.~Sun, ``Sampling and remote estimation for the
  {Ornstein-Uhlenbeck} process through queues: Age of information and beyond,''
  \emph{IEEE/ACM Trans. Netw.}, vol.~29, no.~5, p. 1962–1975, oct 2021.

\bibitem{Wiener_TIT}
Y.~{Sun}, Y.~{Polyanskiy}, and E.~{Uysal}, ``Sampling of the {Wiener} process
  for remote estimation over a channel with random delay,'' \emph{IEEE Trans.
  Inf. Theory}, vol.~66, no.~2, pp. 1118--1135, 2020.

\bibitem{Klugel_infocom}
M.~Klügel, M.~H. Mamduhi, S.~Hirche, and W.~Kellerer, ``Ao{I-Penalty}
  minimization for networked control systems with packet loss,'' in \emph{IEEE
  INFOCOM WKSHPS}, 2019, pp. 189--196.

\bibitem{jsac_survey}
R.~D. Yates, Y.~Sun, D.~R. Brown, S.~K. Kaul, E.~Modiano, and S.~Ulukus, ``Age
  of information: An introduction and survey,'' \emph{IEEE J. Sel. Areas
  Commun.}, vol.~39, no.~5, pp. 1183--1210, 2021.

\bibitem{zhong2018two}
J.~Zhong, R.~D. Yates, and E.~Soljanin, ``Two freshness metrics for local cache
  refresh,'' in \emph{IEEE ISIT}, 2018, pp. 1924--1928.

\bibitem{Ali_AoII}
A.~Maatouk, S.~Kriouile, M.~Assaad, and A.~Ephremides, ``The age of incorrect
  information: A new performance metric for status updates,'' \emph{IEEE/ACM
  Trans. Netw.}, vol.~28, p. 2215–2228, oct 2020.

\bibitem{zheng2020urgency}
X.~Zheng, S.~Zhou, and Z.~Niu, ``Urgency of information for context-aware
  timely status updates in remote control systems,'' \emph{IEEE Trans. Wirel.
  Commun.}, vol.~19, no.~11, pp. 7237--7250, 2020.

\bibitem{yates2021age}
R.~D. Yates, ``The age of gossip in networks,'' in \emph{IEEE ISIT}, 2021, pp.
  2984--2989.

\bibitem{holm2021freshness}
J.~Holm, A.~E. Kal{\o}r, F.~Chiariotti, B.~Soret, S.~K. Jensen, T.~B. Pedersen,
  and P.~Popovski, ``Freshness on demand: Optimizing age of information for the
  query process,'' in \emph{IEEE ICC}, 2021, pp. 1--6.

\bibitem{Uysal_QAoI_ISIT}
M.~E. Ildiz, O.~T. Yavascan, E.~Uysal, and O.~T. Kartal, ``Query age of
  information: Optimizing {A}o{I} at the right time,'' in \emph{IEEE ISIT},
  2022, pp. 144--149.

\bibitem{xiong2022index}
G.~Xiong, X.~Qin, B.~Li, R.~Singh, and J.~Li, ``Index-aware reinforcement
  learning for adaptive video streaming at the wireless edge,'' in \emph{ACM
  MobiHoc}, 2022, pp. 81--90.

\bibitem{zou2021minimizing}
Y.~Zou, K.~T. Kim, X.~Lin, and M.~Chiang, ``Minimizing age-of-information in
  heterogeneous multi-channel systems: A new partial-index approach,'' in
  \emph{ACM MobiHoc}, 2021, pp. 11--20.

\bibitem{Pull_Ji}
F.~Li, Y.~Sang, Z.~Liu, B.~Li, H.~Wu, and B.~Ji, ``Waiting but not aging:
  Optimizing information freshness under the pull model,'' \emph{IEEE/ACM
  Trans. Netw.}, vol.~29, no.~1, pp. 465--478, 2021.

\bibitem{Ornee_SPAWC}
T.~Z. Ornee and Y.~Sun, ``Performance bounds for sampling and remote estimation
  of {Gauss-Markov} processes over a noisy channel with random delay,'' in
  \emph{IEEE SPAWC}, 2021, pp. 1--5.

\bibitem{kadota2018scheduling}
I.~Kadota, A.~Sinha, E.~Uysal-Biyikoglu, R.~Singh, and E.~Modiano, ``Scheduling
  policies for minimizing age of information in broadcast wireless networks,''
  \emph{IEEE/ACM Trans. Netw.}, vol.~26, no.~6, pp. 2637--2650, 2018.

\bibitem{hsu2018age}
Y.-P. Hsu, ``Age of information: Whittle index for scheduling stochastic
  arrivals,'' in \emph{IEEE ISIT}, 2018, pp. 2634--2638.

\bibitem{IT_Cover}
T.~M. Cover and J.~A. Thomas, \emph{Elements of Information Theory (Wiley
  Series in Telecommunications and Signal Processing)}.\hskip 1em plus 0.5em
  minus 0.4em\relax USA: Wiley-Interscience, 2006.

\bibitem{2015ISITYates}
R.~Yates, ``Lazy is timely: Status updates by an energy harvesting source,'' in
  \emph{IEEE ISIT}, 2015.

\bibitem{dawid1998coherent}
A.~P. Dawid, ``Coherent measures of discrepancy, uncertainty and dependence,
  with applications to {Bayesian} predictive experimental design,''
  \emph{Department of Statistical Science, University College London}, vol.
  139, 1998.

\bibitem{farnia2016minimax}
F.~Farnia and D.~Tse, ``A minimax approach to supervised learning,''
  \emph{NeurIPS}, vol.~29, 2016.

\bibitem{whittle_restless}
P.~Whittle, ``Restless bandits: activity allocation in a changing world,''
  \emph{Journal of Applied Probability}, vol. 25A, pp. 287--298, 1988.

\bibitem{weber1990index}
R.~R. Weber and G.~Weiss, ``On an index policy for restless bandits,''
  \emph{Journal of applied probability}, vol.~27, no.~3, pp. 637--648, 1990.

\bibitem{kosta2017age}
A.~Kosta, N.~Pappas, A.~Ephremides, and V.~Angelakis, ``Age and value of
  information: Non-linear age case,'' in \emph{2017 IEEE International
  Symposium on Information Theory (ISIT)}.\hskip 1em plus 0.5em minus
  0.4em\relax IEEE, 2017, pp. 326--330.

\bibitem{sun2017update}
Y.~Sun, E.~Uysal-Biyikoglu, R.~D. Yates, C.~E. Koksal, and N.~B. Shroff,
  ``Update or wait: How to keep your data fresh,'' \emph{IEEE Transactions on
  Information Theory}, vol.~63, no.~11, pp. 7492--7508, 2017.

\bibitem{verloop2016asymptotically}
I.~M. Verloop, ``Asymptotically optimal priority policies for indexable and
  nonindexable restless bandits,'' 2016.

\bibitem{nedic2008subgradient}
A.~Nedic and A.~Ozdaglar, ``Subgradient methods in network resource allocation:
  Rate analysis,'' in \emph{IEEE CISS}, 2008, pp. 1189--1194.

\bibitem{gast2023linear}
N.~Gast, B.~Gaujal, and C.~Yan, ``Linear program-based policies for restless
  bandits: Necessary and sufficient conditions for (exponentially fast)
  asymptotic optimality,'' \emph{Mathematics of Operations Research}, 2023.

\end{thebibliography}


 




\vfill

\end{document}